\documentclass[10pt,aps,pra,onecolumn,superscriptaddress,nofootinbib]{revtex4-2}
\usepackage[T1]{fontenc}
\usepackage{amsmath,amsfonts,color}
\usepackage{gnuplottex}
\usepackage{dsfont}
\usepackage{graphicx}
\usepackage[colorlinks]{hyperref}
\usepackage{bbold}
\newtheorem{theorem}{Theorem}

\newtheorem{corollary}{Corollary}

\newtheorem{lemma}{Lemma}

\newtheorem{proposition}{Proposition}
\newtheorem{remark}{Remark}

\newenvironment{proof}[1][Proof]{\noindent\textbf{#1.} }{\ \rule{0.5em}{0.5em}}

\def\freq{\operatorname{freq}}
\def\Tr{\operatorname{Tr}}
\def\Pr{\operatorname{Pr}}

\def\supp{\operatorname{supp}}

\def\>{\rangle}
\def\<{\langle}
\def\({\left(}
\def\){\right)}

\def\id{\operatorname{id}}

\newcommand{\ket}[1]{\left|{#1}\right\rangle}
\newcommand{\bra}[1]{\left\langle{#1}\right|}
\newcommand{\pro}[1]{\ket{#1}\bra{#1}}

\newcommand{\mc}[1]{\mathcal{#1}}

\usepackage{float}
\usepackage{xcolor}

\begin{document}

\title{Improved finite-size key rates for discrete-modulated continuous variable quantum key distribution under coherent attacks}

\author{Carlos Pascual-Garc\'ia}
\email{carlos.pascual@icfo.eu}
\affiliation{ICFO-Institut de Ciencies Fotoniques, The Barcelona Institute of Science and Technology, Av. Carl Friedrich Gauss 3, 08860 Castelldefels (Barcelona), Spain}

\author{Stefan B{\"a}uml}
\affiliation{ICFO-Institut de Ciencies Fotoniques, The Barcelona Institute of Science and Technology, Av. Carl Friedrich Gauss 3, 08860 Castelldefels (Barcelona), Spain}

\author{Mateus Araújo}
\affiliation{Departamento de Física Teórica, Atómica y Óptica, Universidad de Valladolid, 47011 Valladolid, Spain}

\author{Rotem Liss}
\affiliation{ICFO-Institut de Ciencies Fotoniques, The Barcelona Institute of Science and Technology, Av. Carl Friedrich Gauss 3, 08860 Castelldefels (Barcelona), Spain}

\author{Antonio Ac\'in}
\affiliation{ICFO-Institut de Ciencies Fotoniques, The Barcelona Institute of Science and Technology, Av. Carl Friedrich Gauss 3, 08860 Castelldefels (Barcelona), Spain}
\affiliation{ICREA - Instituci\'{o} Catalana de Recerca i Estudis Avan\c{c}ats, 08010 Barcelona, Spain.}
\date{\today}

\begin{abstract}
Continuous variable quantum key distribution (CVQKD) with discrete modulation combines advantages of CVQKD, such as the implementability using readily available technologies, with advantages of discrete variable quantum key distribution, such as easier error correction procedures. We consider a prepare-and-measure CVQKD protocol, where Alice chooses from a set of four coherent states and Bob performs a heterodyne measurement, the result of which is discretised in both key and test rounds. We provide a security proof against coherent attacks in the finite-size regime based on the generalised entropy accumulation theorem. We compute the corresponding asymptotic key rate imposing a photon-number cutoff and recent advances in conic optimisation. Our method yields significant improvements on key rates: at metropolitan distances, it provides positive key rates for  block sizes of the order of $10^8$ rounds, orders of magnitude smaller compared to previous works.
\end{abstract}

\maketitle
\tableofcontents

\section{Introduction}
Quantum key distribution (QKD) allows two honest parties, conventionally named Alice and Bob, to derive a secret key the security of which is guaranteed by the principles of quantum physics. The first QKD protocols were \emph{discrete variable} (DV) protocols \cite{BB84,E91,bennett1992quantum}
which typically required the legitimate parties to employ qubits. More recently, QKD protocols based on \emph{continuous variable} (CV) states, such as squeezed or coherent states
\cite{cerf2001quantum,grosshans2002continuous,weedbrook2004quantum,garcia2009continuous} have been developed, which can be implemented using readily available optical components and currently existing telecom infrastructure. Hence,
continuous variable QKD  (CVQKD) protocols present several
advantages over their discrete variable QKD (DVQKD) counterparts, including a simpler, affordable implementation that provides optimal scalability, especially at metropolitan  scales \cite{zhang2024continuousvariable,diamanti2015distributing}.

In contrast to DVQKD, where a diverse range of security proofs exist \cite{shor2000simple,gottesman2004security,renner2005information,renner2008security,koashi2009simple}, the scope of security analysis for CVQKD is mostly focused on the case of protocols based on a \textit{Gaussian modulation} \cite{weedbrook2004quantum}. Such an approach can be applied to optical quantum states
(e.g. coherent or squeezed states) and it has been proven to be secure against both collective \cite{navascues2006optimality,garcia2006unconditional,leverrier2010simple,leverrier2015composable} and general attacks \cite{leverrier2013security,leverrier2015composable,leverrier2017security,furrer2012continuous,furrer2014reverse}.

Unfortunately, implementing CVQKD
protocols with Gaussian modulation is challenging
because they require the exact implementation of the Gaussian sampling over all coherent states and complex error correction schemes.
For this reason, in this work, we focus on
\emph{discrete-modulated} CVQKD (DM CVQKD) protocols, where Alice employs a typically small discrete set of coherent states according to a pre-determined distribution, and then Bob applies a coarse-grained measurement. Such protocols provide a better scalability and reduced energy consumption, since they employ simpler error correction  and digital signal processing algorithms \cite{leverrier2009unconditional}. The first DM CVQKD protocols have already been implemented experimentally \cite{pan2022experimental}. On the other hand, DM CVQKD protocols are hindered by the fact that their security against general attacks in practical scenarios is not yet completely understood. The asymptotic and finite-block-size security of such protocols has been shown in the case of collective Gaussian attacks \cite{papanastasiou2019continuous}. In the case of general collective attacks, there are a number of works proving security both in the asymptotic \cite{kaur2019asymptotic,ghorai2019asymptotic,lin2019asymptotic,upadhyaya2021dimension,liu2021homodyne,denys2021explicit} and finite-block-size \cite{lupo_quantum_2022,kanitschar2023finite} settings.

Recent results also provide composable, finite-size security proofs against coherent attacks for binary modulated protocols using homodyne and heterodyne detection \cite{matsuura2021finite,Matsuura2023,yamano2024finite}, as well as protocols using four coherent states and heterodyne detection \cite{BPWFA23}. The former works are based on Azuma's inequality \cite{azuma1967weighted}, whereas the latter uses the original version of the entropy accumulation theorem (EAT) \cite{dupuis2016entropy,dupuis2019entropy}, applied to an entanglement-based version of the protocol. When applying the proofs to a prepare-and-measure scheme, however, some challenges arise, requiring a virtual tomographic measurement on Alice's side, which lead to decreased non-asymptotic rates.

In this work, we derive a composable security proof against general attacks for a DM CVQKD protocol using four coherent states (also known as 4-PSK protocol) and discretised heterodyne measurements, and provide a finite-size key analysis. The proof is related to the one provided in \cite{BPWFA23},
although it has two main important differences: first, the finite-size analysis is performed via
the generalised entropy accumulation theorem (GEAT)
\cite{metger2022generalised,metger2022security}, which can be directly applied to prepare-and-measure protocols without relying on a virtual tomography, providing improved finite-size secret key rates. Second, for the calculation of the asymptotic key rates, for which we impose a photon number cutoff~\cite{lin2019asymptotic,BPWFA23}, we use new numerical techniques from conic optimization \cite{skajaa2015,lorente2024,he2024} that provide an enhanced, more reliable performance in contrast with previous approaches based on the Frank-Wolfe algorithm \cite{winick2018reliable}. Our analysis allows us to distill secret key rates for  $n\sim10^8$ rounds at metropolitan distances.

\section{Preliminaries}
\subsection{Basic notations}
Let us introduce some notation. For a Hilbert space $\mc{H}_A$, $\mc{D}(\mc{H}_A)$ denotes the set of density matrices, i.e. positive semidefinite matrices $\rho_A$ such that $\Tr[\rho_A]=1$, acting on quantum system $A$. Unless otherwise stated, we allow the Hilbert space to be infinite-dimensional, but assume it to be separable\footnote{In the case of the radiation field, this property is guaranteed by the existence of the Fock states, since they constitute an orthonormal countable basis in the corresponding Hilbert space.}, such that any infinite-dimensional state can be arbitrarily well-approximated by a finite-rank one. For convenience, we will occasionally consider subnormalised states. I.e. $\Tr[\rho]\leq1$, and for such cases, we define the set $\mc{D}_{\leq}(\mc{H}_A)$.  By $\mc{H}_{AB}$ we denote a tensor product Hilbert space $\mc{H}_A\otimes\mc{H}_B$, and by $\rho_{AB}$ a bipartite density matrix on it. We further express classical random variables $X$, taking values $\{x\}$ according to the distribution $\{p_x\}$ as density matrices  $\rho_X=\sum_xp_x\pro{x}_X$. By $XY$ we denote the Cartesian product of two random variables $X$ and $Y$. We further use the notation $A_1^n=A_1A_2\ldots A_n$ and $X_1^n=X_1X_2\ldots X_n$ for quantum and classical systems. Classical-quantum (cq) states are expressed using the notation $\rho_{XA}=\sum_xp_x\pro{x}_X\otimes\rho_A^x$. For a cq state $\rho_{CQ} = \sum_{c} p(c) \pro{c} \otimes \rho_c$, we define an event $\Omega$ as a subset of the elements $\{c\}$. The non-normalised conditional state is given by $\rho_{CQ\wedge\Omega} = \sum_{c\in\Omega} p(c) \pro{c} \otimes \rho_c$ and the normalised 
conditional state is then given by $\rho_{CQ}|_{\Omega} =\frac{1}{\Pr[\Omega]} \rho_{CQ\wedge\Omega}$ where $\Pr[\Omega] := \sum_{c\in\Omega} p(c)$. 

For two subnormalised states $\rho,\sigma\in\mc{D}_{\leq}(\mc{H}_A)$, we define the generalised fidelity as
\begin{equation}
    F(\rho,\sigma)=\(\Tr\left|\sqrt{\rho}\sqrt{\sigma}\right|+\sqrt{(1-\Tr[\rho])(1-\Tr[\sigma])}\)^2,
\end{equation}
from which the purified distance follows as
\begin{equation}
    P(\rho,\sigma)=\sqrt{1-F(\rho,\sigma)}.
\end{equation}
Further, the generalised trace distance is defined as
\begin{equation}
    \Delta(\rho,\sigma)=\frac{1}{2}\|\rho-\sigma\|_1+\frac{1}{2}\left|\Tr[\rho-\sigma]\right|.
\end{equation}
The generalised trace distance and the purified distance are metrics on $\mc{D}_{\leq}(\mc{H}_A)$. Both are related by the Fuchs van de Graaf inequality
\begin{equation}\label{Fuchs}
  \Delta(\rho,\sigma)\leq  P(\rho,\sigma)\leq\sqrt{2\Delta(\rho,\sigma)-\Delta(\rho,\sigma)^2}\leq\sqrt{2\Delta(\rho,\sigma)}.
\end{equation}
In this work we will make use of a number of entropic quantities. In addition to the well-known von Neumann entropy $H(\rho)=-\Tr[\rho\log\rho]$, the conditional von Neumann entropy $H(A|B)_{\rho}=H(AB)_{\rho}-H(B)_{\rho}$, as well as the Umegaki relative entropy,
\begin{equation}
    D(\rho||\sigma)=\begin{cases}
    \frac{1}{\Tr[\rho]}\Tr\left[\rho(\log\rho-\log\sigma)\right]&\text{ if }\supp(\rho)\subset\supp(\sigma)\\
    \infty&\text{ otherwise},
    \end{cases}
    \end{equation}
for positive semidefinite $\rho$ and $\sigma$, we will make use of the conditional min entropy, defined for a subnormalised quantum state $\rho_{AB}\in\mc{D}_{\leq}(\mc{H}_{AB})$ by \cite{tomamichel2015quantum}
\begin{align}
    &H_{\min}(A|B)_\rho=\sup_{\sigma_B\in\mc{D}_{\leq}(\mc{H}_{B})}\sup\left\{\lambda\in\mathbb{R}:\rho_{AB}\leq\exp(-\lambda)\mathbb{1}_A\otimes\sigma_B\right\},\label{def:minentropy}
\end{align}
For $\epsilon\geq0$, we can then define the smoothed version of the min entropy as \cite{berta2016smooth}
\begin{align}
    &H^\epsilon_{\min}(A|B)_\rho=\max_{\bar{\rho}\in\mc{B}^\epsilon(\rho_{AB})}H_{\min}(A|B)_{\bar{\rho}},\label{def:smoothmin}
\end{align}
where $\mc{B}^\epsilon(\rho_{A})$ is the $\epsilon$-ball around a state $\rho_A$ in terms of purified distance, i.e. the set of subnormalised states $\tau\in\mc{D}_{\leq}(\mc{H}_{A})$ such that $P(\tau,\rho)\leq\epsilon$. 

\subsection{Security definition}

When two parties, Alice and Bob, wish to communicate in perfect secrecy in the presence of a quantum eavesdropper Eve, they can perform a QKD protocol, typically consisting of $n$ rounds of quantum communication and local measurements, followed by classical post-processing steps involving parameter estimation, error correction and privacy amplification. An instance of a QKD protocol may be aborted if any steps such as parameter estimation fail, or if a subroutine such as error correction aborts. If the protocol does not abort, the goal is to obtain a state close to a fully secret ccq state of the form $\rho^\text{ideal}_{K_AK_BE}=\frac{1}{d}\sum_{x=0}^{d-1}\pro{xx}_{K_AK_B}\otimes\rho_E$, where Alice and Bob's systems are classical, while Eve's system may be quantum. Such a state corresponds to $\log d$ bits of an ideal classical key between Alice and Bob which is secret in that it is completely uncorrelated from Eve. Furthermore, it is correct in the sense that Alice and Bob's systems are perfectly classically correlated.

A security proof for a QKD protocol then involves two parts: Firstly, it has to be shown that the protocol results in a state that is secret and correct, i.e. close to a fully secret ccq-state. Secrecy and correctness combined are also called \emph{soundness}. Formally,  for $\epsilon^\mathrm{sou}>0$, a QKD protocol is said to be \emph{$\epsilon^\mathrm{sou}$-sound} if it results in a state $\rho^\mathrm{QKD}_{ABE}$ after the classical post-processing is performed such that, if we condition on the event ${\Omega_\text{NA}}$ of not aborting the protocol, it holds
\begin{equation}
\Pr[\Omega_{\mathrm{NA}}]\frac{1}{2}\left\|\rho^\mathrm{QKD}_{K_AK_BE}|_{\Omega_\text{NA}}-\rho^\text{ideal}_{K_AK_BE}\right\|_1\leq\epsilon^\mathrm{sou}.
\end{equation}
As we wish to treat the error correction protocol separately from the the remaining protocol, it will be convenient to split the soundness property into a secrecy and correctness part. Namely, let $\epsilon^\mathrm{sec}>0$ and $\epsilon^\mathrm{cor}>0$. A QKD protocol is said to be \emph{$\epsilon^\mathrm{sec}$-secret} if 
\begin{equation}\label{def:secret}
\Pr[\Omega_{\mathrm{NA}}]\frac{1}{2}\left\|\rho^\mathrm{QKD}_{K_BE}|_{\Omega_\text{NA}}-\rho^\text{ideal}_{K_BE}\right\|_1\leq\epsilon^\mathrm{sec}.
\end{equation}
The protocol is further said to be \emph{$\epsilon^\mathrm{cor}$-correct} if
\begin{equation}
    \Pr[K_A\neq K_B\land\Omega_\mathrm{NA}]\leq \epsilon^\mathrm{cor}.
\end{equation}
If the protocol is both $\epsilon^\mathrm{sec}$-secret and $\epsilon^\mathrm{cor}$-correct, it can be shown by means of the triangle inequality for the generalised trace distance that it is $(\epsilon^\mathrm{sec}+\epsilon^\mathrm{cor})$-sound as well. The second part of a security proof is to show completeness. Namely, that there is an honest implementation (i.e. an implementation without presence of Eve) that does not abort with very high probability. Formally, for $\epsilon^\mathrm{com}>0$, we say that a QKD protocol is \emph{$\epsilon^\mathrm{com}$-complete}, if
\begin{equation}
1-\Pr[\Omega_{\mathrm{NA}}]|_\mathrm{hon}\leq\epsilon^\mathrm{com},
\end{equation}
where the subscript $\mathrm{hon}$ denotes an honest implementation.

\section{The protocol}\label{sec:Protocol}

We consider a variant of the 4-PSK protocol employing heterodyne measurements, as described in \cite{BPWFA23}.
For each round, Alice randomly chooses one out of four possible coherent states
$\ket{\varphi_x} \in \{\ket{i^x \alpha}\}_{x=0}^3$, where the amplitude
$\alpha \in \mathbb{R}$ is an agreed parameter of the protocol;
Alice then sends the state to Bob through an insecure
quantum channel, and Bob performs a coarse-grained heterodyne measurement, discretising the signal according to a pre-determined binning. Our approach differs from the one of \cite{lin2019asymptotic} in two aspects---our protocol does not include post-selection, and we apply a discretisation for Bob's measurement outputs in both key and parameter estimation (test) rounds, and not only in key rounds.

Let us formalise our protocol as
follows: Alice and Bob first agree on the total number of rounds
$n \in \mathbb{N}$, a modulation given by the design parameters $\delta,\Delta$, such that $0<\delta<\Delta$, as well as the amplitude for the states $\alpha \in \mathbb{R}$. For convenience, let us define $p^\mathrm{K},p^\mathrm{PE}\in [0,1]$ such that $p^\mathrm{PE} = 1 - p^\mathrm{K}$ which will denote the probabilities that a round is used for key distillation or parameter estimation. For every round $j \in \{1,\ldots,n\}$, Alice and Bob perform the following operations:

\begin{enumerate}
\item \label{protocol_prep}
{{\it Alice's State Preparation:}
Alice chooses a uniformly random value $x_j\in\{0,\ldots,3\}$, corresponding to a random variable $X_j$, defining her prepared quantum state
$\ket{\varphi_{x_j}}_{A'}$. She then sends the state $\ket{\varphi_{x_j}}_{A'}$ to Bob
via an insecure quantum channel controlled by Eve. 
Alice also randomly chooses the kind of round via a random variable $I_j$, taking $I_j=0$ for key rounds and $I_j=1$ for test rounds (i.e. used for or parameter estimation), whose respective probabilities are $\Pr[I_j=0]=p^\mathrm{K}$ and $\Pr[I_j=1]=p^\mathrm{PE}$.
For the following discussion, we define separate random variables corresponding to the two types of rounds:
\begin{align}
\hat{X}_j&=\begin{cases}
x_j & \text{ if }I_j=0,\\
\perp &\text{ else}.
\end{cases}\\
\tilde{X}_j&=\begin{cases}
x_j & \text{ if }I_j=1,\\
\perp & \text{ else}.
\end{cases}
\end{align}
After Bob receives and measures the quantum state,
Alice sends to Bob the value of the random variable $I_j$
via a classical authenticated channel.
}

\item \label{protocol_meas} {{\it Bob's Measurement:}
Bob performs a heterodyne measurement on subsystem $B_j$ which he has received from Alice via the quantum channel, yielding a continuous random variable $Y_j$ which takes complex values $y_j\in\mathbb{C}$. Again, let us conveniently define separate random variables corresponding to the two types of rounds:
\begin{align}
\hat{Y}_j&=\begin{cases}
y_j &\text{ if }I_j=0,\\
\perp &\text{ else}.
\end{cases}\\
\tilde{Y}_j&=\begin{cases}
y_j & \text{ if }I_j=1,\\
\perp & \text{ else}.
\end{cases}
\end{align}}

\item \label{protocol_discr} {{\it Discretisation:}
Bob translates his continuous outcome
of the heterodyne measurement into a discrete result. This process is different for key and parameter estimation rounds---for key rounds, Bob discretises
the result according to a division of the phase space into quadrants. Formally, denoting $\hat{y}_j = |\hat{y}_j|e^{i\hat{\theta}_j}$,
Bob defines the discrete random variable for key rounds as
\begin{equation}
\hat{Z}_j=\begin{cases}
z &\text{ if }{ \hat{\theta}_j\in[\frac{\pi}{4}(2z-1),\frac{\pi}{4}(2z+1))} \wedge I_j= 0,\\
\perp &\text{ else }
\end{cases}
\end{equation}
for $z \in \{ 0,\ldots,3\}$ and $\hat{Z}_j=\perp$ explicitly representing the non-key rounds. In test rounds, Bob performs a discretisation by separating the phase space into six modules as follows: let $\delta$ and $\Delta$ with  $0<\delta<\Delta$; for amplitudes smaller than $\delta$, the phase space is separated into four modules as in key rounds. Outcomes whose amplitudes lie between $\delta$ and $\Delta$ correspond to a fifth module, and amplitudes greater or equal than $\Delta$ correspond to the sixth module (see also Figure \ref{fig:PhaseSpaceModules}). Denoting $\tilde{y}_j = |\tilde{y}_j|
e^{i\tilde{\theta}_j}$ and $z \in \{0,\ldots,3 \}$, we define
\begin{equation}\label{eq:discPE}
\tilde{Z}_j=\begin{cases}
z &\text{ if }{ \tilde{\theta}_j\in[\frac{\pi}{4}(2z-1),\frac{\pi}{4}(2z+1))\land  |\tilde{y}_j| \in [0,\delta)} \wedge I_j= 1 ,\\
4 &\text{ if } {|\tilde{y}_j|\in [\delta,\Delta)} \wedge I_j= 1 , \\
5 &\text{ if } {|\tilde{y}_j|\in [\Delta,\infty)}\wedge I_j= 1, \\
\perp &\text{ else.}
\end{cases}
\end{equation}

}
\end{enumerate}

Summarising steps \ref{protocol_prep}--\ref{protocol_discr},
each round $i$ of the protocol creates discrete classical
random variables $\hat{X}_i$ and $\hat{Z}_i$
for key generation, as well as $\tilde{X}_i$ and $\tilde{Z}_i$
for parameter estimation. 
We denote $C_i = \tilde{X}_i\tilde{Z}_i$, 
as the register holding all 
the classical information
used for parameter estimation in round $i$, and $C_1^n = C_1\ldots C_n$ for all said registers collected during the execution of the protocol. 
Let us define Eve's side information for rounds $1$ to $i$ as $E'_i$
(so that $E'_i$ includes all information
from previous rounds, $E'_1,\ldots,E'_{i-1}$), and we note that Eve has
access to registers $I_i$ and $\tilde{X}_i$ sent by Alice to Bob. 
All in all, the joint quantum state of Alice, 
Bob and Eve after $n$ rounds of steps \ref{protocol_prep}--\ref{protocol_discr} is thus
$\sigma^\mathrm{QKD}_{\hat{X}_1^n\hat{Z}_1^nC_1^nI_1^nE'_n}$. 
We note also that the round $i+1$ only starts after the $i$-th signal has arrived at Bob's laboratory. This is used to ensure that Eve only possesses one transmitted state at the same time, which is needed for applying the GEAT to prepare-and-measure protocols using the method of \cite{metger2022security} (cf. \cite[Condition 3.1]{metger2022security}).

Next, Bob performs parameter estimation,
which he does by comparing the test results $C_1^n$ with the probability distribution representing
the test results in an ideal, honest implementation. 
If $C_1^n$ deviates excessively from
the results expected in an ideal run, 
the protocol is aborted. 
Let us formalise this notion by considering 
the string of test results
$c_1^n = (c_1,...,c_n)$, where $c_i$
takes values from an alphabet $\mathcal{C}$---if round $i$ is a key round, then $c_i=(\perp,\perp)$
whereas $c_i = (x_i, z_i)$ for parameter estimation
rounds, for all $x_i\in\{0,1,2,3\}$ and
$z_i\in\{0,1,2,3,4,5\}$. 
Thus, we define $\mc{C}$
as
\begin{equation}\label{def:alphabetC}
    \mc{C} = \tilde{\mc{C}}\cup \{(\perp,\perp)\},
\end{equation}
where $\tilde{\mc{C}}$ is the set of possible values
for $c_i$ in parameter estimation rounds:
\begin{equation}
    \tilde{\mc{C}} = \{0,1,2,3\} \times \{0,1,2,3,4,5\}.
\end{equation}
To compare the observed outcomes $C_1^n$ 
with the ideal expected ones,
we define a probability distribution $p_0\in\mc{P}_{\mc{C}}$ on the alphabet $\mc{C}$ to represent the results of an 
honest run of the protocol.
This distribution is identical and independent for each round. Next, we split it according to the 
simulated probabilities $p_0^\mathrm{s}$ of
a parameter estimation round, happening 
with probability $p^\mathrm{PE}$.
Distributions $p_0$ and $p_0^\mathrm{s}$ are thus related
as
\begin{eqnarray}
    p_0(x,z) &=& p^\mathrm{PE}p^\mathrm{s}_0(x,z), \quad 
    \forall (x,z) \in \tilde{\mathcal{C}},\label{eq:p0PE}\\
   p_0(\perp,\perp) &=& 1-p^\mathrm{PE}
   = p^\mathrm{K}.\label{eq:p0PE2}
\end{eqnarray}
We will provide the explicit form
of $p_0^\mathrm{s}$ in Section \ref{subs:NumImplementation}. Now, let us compare the observed frequency distribution of the test results $ c_1^n \in \mc{C}^n$ to the expected probability distribution $p_0$. Given a string $c_1^n$ of test results, we denote by $\freq_{c_1^n}\in\mc{P}_\mc{C}$
the frequency distribution corresponding 
to each possible value in $c_1^n$.
Thus, we define for any $c\in\mc{C}$
\begin{equation}
\freq_{c_1^n}(c)=\frac{|\{i \in \{1, 2, \ldots, n\}
: c_i=c\}|}{n}.
\end{equation}
Given these conditions, it is possible to 
compare $\freq_{c_1^n}$
and $p_0$ by means of an affine function
$f^\text{PE}:\mc{P}_{\mc{C}}\to\mathbb{R}$ which maps any probability distribution
$p\in\mc{P}_{\mc{C}}$ to a real number $\mathbb{R}$. Such a function may be written as
\begin{equation}
    f^\mathrm{PE}(p)
    = \sum_{(x,z)\in\tilde{\mathcal{C}}} h_{xz} p(x,z),
    \label{eq:fPE}
\end{equation}
for some scalar coefficients
$h_{x,z}\in\mathbb{R}$, whose explicit form will be provided later in \ref{Subsec:Cutoff}. We will also see that the function $f^\mathrm{PE}$ is close to what, in the context of GEAT, is referred to as a min-tradeoff function (see eq. (\ref{eq:MinTrade}) for an exact definition). Bob then evaluates the affine function $f^\mathrm{PE}$ on both probability distributions $\freq_{c_1^n}$ and $p_0$, and he aborts the protocol if and only if $f^\mathrm{PE}(\freq_{c_1^n}) < f^\mathrm{PE}(p_0)
- \delta^\mathrm{tol}$ for a given deviation tolerance $\delta^\mathrm{tol}>0$. In other words, the event of passing the parameter estimation
without aborting the protocol is defined as
\begin{equation}\label{eq:OmegaPE}
    \Omega_\mathrm{PE}:=\left\{c_1^n\in\mc{C}^n
    : f^\mathrm{PE}(\freq_{c_1^n}) \geq f^\mathrm{PE}(p_0)
    - \delta^\mathrm{tol}\right\}.
\end{equation}

If the protocol does not abort
after parameter estimation,
Alice and Bob perform error correction of their
respective raw keys $\hat{X}_1^n$ and $\hat{Z}_1^n$.
This process is performed in terms of a
reverse information reconciliation,
where Bob sends an error correction string (denoted $L$)
to Alice, who then produces her guess for the raw key,
$\bar{X}_1^n$.
To make sure that Alice's raw key $\bar{X}_1^n$
is identical to Bob's raw key $\hat{Z}_1^n$,
Bob chooses a random hash function $F$, computes the hash of
his key $\hat{F} = F(\hat{Z}_1^n)$,
and sends to Alice a description of $F$ and the value $\hat{F}$.
Alice then computes $\bar{F} = F(\bar{X}_1^n)$. If $F(\hat{Z}_1^n) \neq F(\bar{X}_1^n)$,
she aborts the protocol.
The event of passing the error correction step is then defined as
\begin{equation}\label{eq:OmegaEC}
    \Omega_{\mathrm{EC}} = \big[F(\hat{Z}_1^n) = F(\bar{X}_1^n)\big].
\end{equation}
Let us also assume that the probability of mistakenly passing the error correction step is upper bounded by $\epsilon_\mathrm{EC}>0$, 
\begin{equation}
\Pr\big[F(\hat{Z}_1^n) = F(\bar{X}_1^n)\wedge\hat{Z}_1^n\neq \bar{X}_1^n \big]\leq \epsilon_\mathrm{EC}.
\end{equation}
Furthermore, we consider that the probability of
\emph{not} passing the error correction step in
an ideal run (without an attack)
is upper bounded by $\epsilon^c_\mathrm{EC}>0$, formally described as
\begin{equation}
\Pr\big[F(\hat{Z}_1^n) \neq F(\bar{X}_1^n)\big]\big|_\mathrm{hon}\leq \epsilon^c_\mathrm{EC}.
\end{equation}
To conclude, we define the total event
of nonabortion as succeeding in both parameter estimation and error correction
\begin{equation}\label{eq:OmegaNonAbort}
\Omega_{\mathrm{NA}} = \Omega_{\mathrm{PE}}
\wedge \Omega_{\mathrm{EC}}.
\end{equation}  
We note that $\Omega_{\mathrm{PE}}$ only depends on
$C_1^n$ whereas $\Omega_{\mathrm{EC}}$ depends
on registers $\bar{X}_1^n \hat{Z}_1^n\bar{F}\hat{F}$. Taking all these considerations into account, the final joint state of Alice, Bob, and Eve after the protocol is given by
\begin{equation}\label{RawKeyState}
    \sigma^\mathrm{QKD}_{\bar{X}_1^n \hat{Z}_1^n  C_1^nI_1^n \bar{F}\hat{F} LE'_n}=\Pr[\Omega_{\mathrm{NA}}]\sigma^\mathrm{QKD}_{\bar{X}_1^n \hat{Z}_1^n  C_1^nI_1^n \bar{F}\hat{F} LE'_n}|_{ \Omega_{\mathrm{NA}}}+(1-\Pr[\Omega_{\mathrm{NA}}])\sigma^\mathrm{QKD}_{\bar{X}_1^n \hat{Z}_1^n  C_1^nI_1^n \bar{F}\hat{F} LE'_n}|_{\neg \Omega_{\mathrm{NA}}}.
\end{equation}

\begin{figure}
    \includegraphics[width=0.8\linewidth]{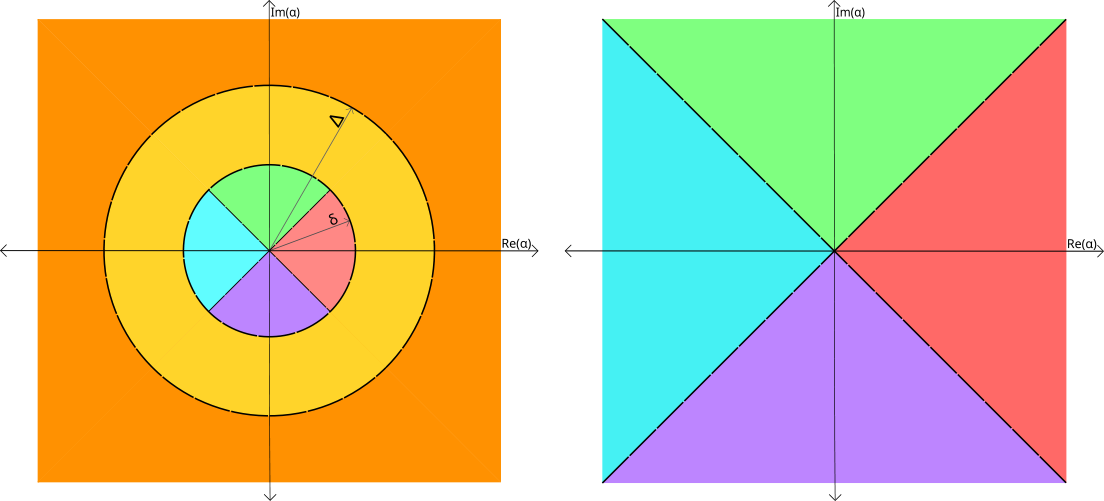}
    \caption{Discretisations of phase space by Bob for parameter estimation (left) based on phases and amplitudes according to the parameters $\delta,\Delta$, and key generation (right) rounds which only require a binning according to the phase.}
    \label{fig:PhaseSpaceModules}
\end{figure}

\begin{remark}
We note that while this work considers a modulation based on six binnings in (\ref{eq:discPE}), our analysis can be easily adapted to any other discretizations. For instance, it is possible to perform an extra modulation for phases also when the amplitude is larger than $\delta$. Namely, for $j \in \{0,\ldots,3\}$
\begin{equation}\label{eq:discPE_2}
\tilde{Z}_i=\begin{cases}
j & \text{ if }{ \tilde{\theta}_i\in[\frac{\pi}{4}(2j-1),\frac{\pi}{4}(2j+1))\land  |\tilde{y}_i| \in [0,\delta)},\\
j + 4 & \text{ if }{ \tilde{\theta}_i\in[\frac{\pi}{4}(2j-1),\frac{\pi}{4}(2j+1))\land  |\tilde{y}_i| \in [\delta, \Delta)},\\
j + 8 & \text{ if } { \tilde{\theta}_i\in[\frac{\pi}{4}(2j-1),\frac{\pi}{4}(2j+1))\land |\tilde{y}_i|\geq \Delta}, \\
\perp & \text{ else}.
\end{cases}
\end{equation}
The schema based on six modules is more suitable for finite block sizes, since its smaller number of binnings decrease the number of rounds needed for parameter estimation, and thus also reduces the overall statistical deviations. On the other hand, a modulation as \eqref{eq:discPE_2} provides better key rates in the asymptotic regime, especially for large distances, albeit it requires larger block sizes to provide positive key rates in the finite regime.
\end{remark}

\section{Security proof}

In this section we show security against coherent attacks. The proof consists of two steps---firstly, in Subsection \ref{sec:sound} we show the soundness of the  protocol and provide a lower bound on the key rate. The soundness proof is based on the GEAT \cite{metger2022generalised,metger2022security}, a tool allowing us to lower bound the amount of key obtainable in the presence of coherent attacks by reducing the problem to the case of only collective attacks. Secondly, in Subsection \ref{sec:compl} we show the completeness of the protocol, i.e. that there is a nonzero probability of it not being aborted. Finally, in Subsection \ref{sec:MinTrade}, we use techniques based on \cite{winick2018reliable,upadhyaya2021dimension} to numerically compute the collective attack bounds.
 
 \subsection{Soundness}\label{sec:sound}

In this Section we lower bound the achievable key rate $r^\mathrm{K}=\ell/n$, where $\ell$ is the length of the final key (in bits) and $n$ the number of rounds, conditioned on the event $ \Omega_{\mathrm{NA}}$ of the protocol not being aborted. We make use the following version of the leftover hash Lemma, noting that it includes the possibility of infinite dimensional Hilbert spaces---particularly, we use the version derived in \cite{furrer2012security,berta2016smooth} for von Neumann algebras, where infinite dimensional Hilbert spaces represent a special case.

\begin{proposition}\label{Furrer_Hash}\emph{\cite{furrer2012security,berta2016smooth}}
    Let $Z$ and $K$ be finite sets such that $|Z|\geq|K|$, and let $(\mathcal{F},\mathcal{P}_\mathcal{F})$ be a family of two-universal $\{Z,K\}$-hash functions. Let $\omega_{ZE}$ be a sub-normalized cq-state, where the quantum system may be infinite dimensional. Let $\epsilon>0$. It then holds
    \begin{equation}
        \|\omega^\mathcal{F}_{KE}-\frac{1}{|K|}\mathbb{1}_K\otimes\omega_E\|_1\leq\sqrt{|K| 2^{-H_{\min}^\epsilon (Z|E)_\omega}}+4\epsilon,
    \end{equation}
    where $\omega^\mathcal{F}_{KE}=\sum_{f\in\mathcal{F}}p(f)T^f_{Z\to K}\otimes\id_E(\omega_{ZE})$ and the map $T^f_{Z\to K}$ corresponds to the application of a hash function $f$ to system $X$.
\end{proposition}
As a corollary, one obtains a bound on the key length.

\begin{corollary}\label{renner}
Let $\epsilon,\epsilon_\mathrm{PE},\epsilon_\mathrm{EC}\in(0,1)$, and $\Omega_\mathrm{PA}$, $\Omega_\mathrm{EC}$ be the events defined by \eqref{eq:OmegaPE} and \eqref{eq:OmegaEC}. Assume a reverse-reconciliation error correction protocol, such that $\Pr[K_A\neq K_B\land\Omega_\mathrm{EC}]\leq \epsilon_\mathrm{EC}$, during which a finite number $\mathrm{leak}_\mathrm{EC}:=\log|L|$ bits are leaked to Eve. Then, if the protocol does not abort, by applying a family of hash functions as in Proposition \ref{Furrer_Hash} to the raw key state $\sigma^\mathrm{QKD}$ given in eq (\ref{RawKeyState}), Alice and Bob are able to extract a key of length
\begin{equation}
\ell\leq H^{\epsilon}_{\min}(\hat{Z}_1^n|\tilde{X}_1^nI_1^nE'_n)_{\sigma^\mathrm{QKD}{|\Omega_\mathrm{PE}\land \Omega_{\mathrm{EC}}}}-\mathrm{leak}_\mathrm{EC}-2\log\frac{1}{\epsilon_\mathrm{PA}}
\end{equation}
bits, which is $\frac{1}{2}\epsilon_\mathrm{PA}+2\epsilon_\mathrm{EC}$-sound. 
\end{corollary}
\begin{proof}
Setting $|K|=2^\ell$, where $\ell$ is the number of bits in the $K$ register, and using the subnormalised state $\omega=\sigma^\mathrm{QKD}{|\Omega_\mathrm{PE}\land \Omega_{\mathrm{EC}}}$, identifying subsystems $Z=\hat{Z}_1^n$, $K=K_B$ and $E=\tilde{X}_1^nI_1^nE'_nL$, Proposition \ref{Furrer_Hash} implies that
\begin{equation}
\left\|\sigma^{\mathrm{QKD},\mathcal{F}}_{K_B\tilde{X}_1^nI_1^nE'_nL}{|\Omega_\mathrm{PE}\land \Omega_{\mathrm{EC}}}-\frac{1}{2^\ell}\mathbb{1}_{K_B}\otimes\sigma^{\mathrm{QKD}}_{\tilde{X}_1^nI_1^nE'_nL}{|\Omega_\mathrm{PE}\land \Omega_{\mathrm{EC}}}\right\|_1\leq\sqrt{2^{\ell-H_{\min}^\epsilon (\hat{Z}_1^n|\tilde{X}_1^nI_1^nE'_nL)_{\sigma^\mathrm{QKD}|\Omega_\mathrm{PE}\land \Omega_{\mathrm{EC}}}}}+4\epsilon,
\end{equation}
where we use the notation $\sigma|\Omega_\mathrm{PE}\land\Omega_\mathrm{EC}:=\frac{1}{\Pr[\Omega_\mathrm{PE}]}\sigma\land\Omega_\mathrm{NA}$. Hence for $\ell\leq H_{\min}^\epsilon (\hat{Z}_1^n|\tilde{X}_1^nI_1^nE'_nL)_{\sigma^\mathrm{QKD}|\Omega_\mathrm{PE}\land \Omega_{\mathrm{EC}}}-2\log\frac{1}{\epsilon_\mathrm{PA}}$ it holds
\begin{align}
    &\Pr[\Omega_\mathrm{NA}]\frac{1}{2}\left\|\sigma^{\mathrm{QKD},\mathcal{F}}_{K_B\tilde{X}_1^nI_1^nE'_nL}|_{\Omega_{\mathrm{NA}}}-\frac{1}{2^\ell}\mathbb{1}_{K_B}\otimes\sigma^{\mathrm{QKD}}_{\tilde{X}_1^nI_1^nE'_nL}|_{\Omega_{\mathrm{NA}}}\right\|_1\\
&\leq\Pr[\Omega_\mathrm{EC}|\Omega_\mathrm{PE}]\frac{1}{2}\left\|\sigma^{\mathrm{QKD},\mathcal{F}}_{K_B\tilde{X}_1^nI_1^nE'_nL}|_{\Omega_{\mathrm{NA}}}-\frac{1}{2^\ell}\mathbb{1}_{K_B}\otimes\sigma^{\mathrm{QKD}}_{\tilde{X}_1^nI_1^nE'_nL}|_{\Omega_{\mathrm{NA}}}\right\|_1\\
&=\frac{1}{2}\left\|\sigma^{\mathrm{QKD},\mathcal{F}}_{K_B\tilde{X}_1^nI_1^nE'_nL}{|\Omega_\mathrm{PE}\land \Omega_{\mathrm{EC}}}-\frac{1}{2^\ell}\mathbb{1}_{K_B}\otimes\sigma^{\mathrm{QKD}}_{\tilde{X}_1^nI_1^nE'_nL}{|\Omega_\mathrm{PE}\land \Omega_{\mathrm{EC}}}\right\|_1\\
&\leq \frac{1}{2}\epsilon_\mathrm{PA}+2\epsilon,
\end{align}
Hence, the key is $(\frac{1}{2}\epsilon_\mathrm{PA}+2\epsilon)$-secret by eq. (\ref{def:secret}). Since $\Pr[K_A\neq K_B\land\Omega_\mathrm{NA}]\leq\Pr[K_A\neq K_B\land\Omega_\mathrm{EC}]\leq \epsilon_\mathrm{EC}$, the protocol is also $\epsilon_\mathrm{EC}$-correct, and therefore $(\frac{1}{2}\epsilon_\mathrm{PA}+2\epsilon+\epsilon_\mathrm{EC})$-sound. Since $|L|$ is finite, we can use the chain rule of the smooth min entropy, Lemma 4.5.6, as well as Lemma 4.5.7 in \cite{furrer2012security}, and obtain
\begin{equation}
    H_{\min}^\epsilon (\hat{Z}_1^n|\tilde{X}_1^nI_1^nE'_nL)_{\sigma^\mathrm{QKD}|\Omega_\mathrm{PE}\land \Omega_{\mathrm{EC}}}\geq H_{\min}^\epsilon (\hat{Z}_1^n|\tilde{X}_1^nI_1^nE'_n)_{\sigma^\mathrm{QKD}|\Omega_\mathrm{PE}\land \Omega_{\mathrm{EC}}}-\log|L|,
\end{equation}
finishing the proof.
\end{proof}

Corollary \ref{renner}, reduces the soundness part to finding a lower bound on $H_{\min}^\epsilon (\hat{Z}_1^n|\tilde{X}_1^nI_1^nE'_n)_{\sigma^\mathrm{QKD}|\Omega_\mathrm{PE}\land \Omega_{\mathrm{EC}}}$. 
We will now show that conditioning on $\land\Omega_\mathrm{EC}$ can only increase the smooth min entropy. For that purpose we show that Lemma 10 in \cite{tomamichel2017largely} also holds in infinite dimensions.

\begin{lemma}\label{lemma10vN}
    Let $\rho_{ABXY}$ be a sub-normalized cq-state the quantum part of which acts on a possibly infinite dimensional Hilbert space $\mathcal{H}_{AB}$, let $\Omega$ be an event defined on $XY$, and let $\epsilon>0$. Then it holds
    \begin{equation}
        H^\epsilon_{\min}(AX|B)_{\rho\land\Omega}\geq H^\epsilon_{\min}(AX|B)_{\rho}.
    \end{equation}
\end{lemma}
The proof of Lemma \ref{lemma10vN} follows the one given in \cite{tomamichel2017largely}, but uses definitions and properties of the relevant quantities as provided in \cite{furrer2012security}.

\begin{proof}
    By definition of the smooth min entropy, eq. (\ref{def:smoothmin}), 
    for every $\delta_1>0$ there exists a subnormalised cq-state $\tilde\rho_{ABX}$ in the $\epsilon$-ball around $\rho_{ABX}$ such that $H_{\min}(AX|B)_{\tilde\rho}\leq H^\epsilon_{\min}(AX|B)_{\rho}\leq H_{\min}(AX|B)_{\tilde\rho}+\delta_1$. Further, by definition of the min entropy, eq. (\ref{def:minentropy}), 
    for every $\delta_2>0$ there exists a state $\sigma_B$ and $\lambda\in\mathbb{R}$ such that $\lambda\mathbb{1}_{AX}\otimes\sigma_B\geq\tilde\rho_{ABX}$ and
$\lambda-\delta_2\leq2^{-H_{\min}(AX|B)_{\tilde\rho}}\leq\lambda$. Hence it holds
\begin{equation}
        \tilde\rho_{ABX}\leq \(2^{-H^\epsilon_{\min}(AX|B)_{\rho}-\delta_1}+\delta_2\)\mathbb{1}_{AX}\otimes\sigma_B .
    \end{equation}
By definition of the purified distance there exist purifications $\ket{\tilde\rho_{ABXC}}$ and $\ket{\rho_{ABXC}}$ such that $P(\tilde\rho_{ABX},\rho_{ABX})=P(\ket{\tilde\rho_{ABXC}},\ket{\rho_{ABXC}})$. Consider now $\rho_{ABXY}$ and its purification $\ket{\rho_{ABXY\tilde{C}}}$. By Lemma 3.2.2 of \cite{furrer2012security} there exists an isometry $V$ such that $V\ket{\rho_{ABXY\tilde{C}}}=\ket{\rho_{ABXC}}$. Hence, since the purified distance is invariant under isometries 
and decreases under completely positive trace non-increasing maps (Lemma 3.4.5 of \cite{furrer2012security}), it holds
\begin{equation}
P(\tilde\rho_{ABX},\rho_{ABX})=P(\ket{\tilde\rho_{ABXC}},\ket{\rho_{ABXC}})=P(\ket{\tilde\rho_{ABXY\tilde C}},\ket{\rho_{ABXY\tilde C}})\geq P(\tilde\rho_{ABXY},\rho_{ABXY}),
\end{equation}
implying that $\tilde\rho_{ABXY}$ is in the $\epsilon$-ball around $\rho_{ABXY}$. Since pinching (i.e. the application of a projective measurement in the computational basis) $Y$ would only decrease the distance between $\tilde\rho_{ABXY}$ and $\rho_{ABXY}$, while leaving $\rho_{ABXY}$ invariant, we can, without loss of generality assume that $\tilde\rho_{ABXY}$ is classical on $Y$. We can now write $\tilde\rho_{ABXY}=\tilde\rho_{ABXY\land\Omega}+\tilde\rho_{ABXY\land\neg\Omega}$, where $\tilde\rho_{ABXY\land\Omega}\geq0$ and $\tilde\rho_{ABXY\land\neg\Omega}\geq0$. Consequently, it holds $\tilde\rho_{ABXY\land\Omega}\leq\tilde\rho_{ABXY}$. Since the partial trace is completely positive, this implies
\begin{equation}
    \tilde\rho_{ABX\land\Omega}\leq\tilde\rho_{ABX}\leq \(2^{-H^\epsilon_{\min}(AX|B)_{\rho}-\delta_1}+\delta_2\)\mathbb{1}_{AX}\otimes\sigma_B.
\end{equation}
Hence, by the definition of the min-entropy, it holds in the limit $\delta_1,\delta_2\to0$,
\begin{equation}
    H^\epsilon_{\min}(AX|B)_{\rho_{ABX}}\leq H_{\min}(AX|B)_{\tilde\rho_{ABX\land\Omega}}.
\end{equation}
By considering the completely positive trace non-increasing map transforming $\rho_{ABXY}$ into $\rho_{ABXY\land\Omega}$ and $\tilde\rho_{ABXY}$ into $\tilde\rho_{ABXY\land\Omega}$, we can again make use of Lemma 3.4.5 of \cite{furrer2012security} to show that 
\begin{equation}
    P\(\rho_{ABX\land\Omega},\tilde\rho_{ABX\land\Omega}\)\leq P\(\rho_{ABXY\land\Omega},\tilde\rho_{ABXY\land\Omega}\)\leq  P\(\rho_{ABXY},\tilde\rho_{ABXY}\)\leq\epsilon.
\end{equation}
By definition of the smooth min-entropy it then holds
\begin{equation}
    H_{\min}(AX|B)_{\tilde\rho_{ABX}\land\Omega}\leq H^\epsilon_{\min}(AX|B)_{\rho_{ABX}\land\Omega},
\end{equation}
finishing the proof.
\end{proof}

 In order to apply the GEAT, let us define an output variable  $O_i:=\hat{Z}_i$ as well as $E_i:=I_1^iC_1^iE'_i$, 
 (recall that $C_i=\tilde{X}_i\tilde{Z}_i$), which can be interpreted as Eve's side information. In fact, with $\tilde{Z}_i$ we give Eve more information than necessary by the protocol. For the entropy accumulation, it is however necessary that the information $C_i$ used for statistical tests is obtainable by means of projective measurements on $O_iE_i$, so $\tilde{Z}_i$ has to be included into either $O_i$ or $E_i$.

Let us now define a set of GEAT channels, i.e. CPTP maps $\{\mc{M}^i_{E_{i-1}\to O_iC_iE_i}\}_{i=1}^n$, such that each map $\mc{M}^i_{E_{i-1}\to O_iC_iE_i}$ includes steps (1) - (3) of the actions by Alice and Bob taken during round $i$, and also includes the attack map $\tilde{\mc{N}}^i_{A_i'E'_{i-1}\to B_iE'_i}$ which is performed by Eve. 
Explicitly, the channels take the form
\begin{equation}\label{M}
\mc{M}^{i}(\cdot)=p^\mathrm{K}\mc{M}^i_\mathrm{K}(\cdot)\otimes\pro{\perp}_{C_i}\otimes\pro{0}_{I_i}+(1-p^\mathrm{K})\mc{M}^i_\text{test}(\cdot)\otimes\pro{\perp}_{O_i}\otimes\pro{1}_{I_i},
\end{equation}
where (omitting identities, and double subsystems),
\begin{align}
\mc{M}_\mathrm{K}^i\(\rho_{E_{i-1}}\)&=\frac{1}{4}\sum_{(x,z)=(0,0)}^{(3,3)}\Tr_{B_i}\left[R^z_{B_i}\(\tilde{\mc{N}}^i_{A'_iE'_{i-1}\to B_iE'_i}\(\varphi^{x}_{A_i}\otimes\rho_{E_{i-1}}\)\)\right]\otimes\pro{z}_{O_i}\label{Mkey}\\
\mc{M}_\mathrm{test}^i\(\rho_{E_{i-1}}\)&=\frac{1}{4}\sum_{(x,z)\in\tilde{\mathcal{C}}}\Tr_{B_i}\left[\tilde{R}^z_{B_i}\(\tilde{\mc{N}}^i_{A'_iE'_{i-1}\to B_iE'_i}\(\varphi^{x}_{A_i}\otimes\rho_{E_{i-1}}\)\)\right]\otimes\pro{xz}_{C_i},\label{Mtest}
\end{align}
and we define the region operators for key rounds as
\begin{equation}\label{R}
R^z_B=\frac{1}{\pi}\int_{0}^{\infty}\int_{\frac{\pi (2z-1)}{4}}^{\frac{\pi (2z+1)}{4}}\gamma\pro{\gamma e^{i\theta}}d\theta d\gamma,
\end{equation}
for $z\in\{0,\ldots,3\}$, and the region operators for test rounds as
\begin{subequations}\label{eq:OperatorsPE}
\begin{align}
     \tilde{R}^z_B &= \frac{1}{\pi}\int_{0}^\delta \int_{\frac{\pi(2z-1)}{4}}^{\frac{\pi(2z+1)}{4}} \gamma\pro{\gamma e^{i\theta}}d\theta d\gamma,  && z = 0,\ldots,3,  \\
     \tilde{R}^z_B &= \frac{1}{\pi}\int_\delta^\Delta \int_0^{2 \pi} \gamma\pro{\gamma e^{i\theta}}d\theta d\gamma,  && z = 4,  \\
     \tilde{R}^z_B &= \frac{1}{\pi}\int_\Delta^\infty \int_0^{2 \pi} \gamma\pro{\gamma e^{i\theta}}d\theta d\gamma,  && z = 5.
\end{align}
\end{subequations}


Each operator represents the fragment of the phase space corresponding to a module $z \in \{0,\ldots,5\}$, defined according to the intervals in \eqref{eq:discPE}. Note that by using the region operators, we have merged Bob's heterodyne measurement and discretisation. By definition, an $n$-fold application of the maps $\mc{M}^i_{E_{i-1}\to O_iC_iE_i}$ results in a state that equals the state of the protocol after steps  (1) - (3) on the marginals $O_1^nC_1^nE_1^n$, namely
\begin{equation}\label{eq:15}
\sigma^\mathrm{QKD}_{O_1^nC_1^nE_n}=\mc{M}^n\circ\ldots\circ\mc{M}^1(\rho_{E_0}),
\end{equation}
where Eve's original state $\rho_{E_0}$ can be chosen trivial. As the event $\Omega_\mathrm{PE}$ is defined on $C_1^n$, it then follows from eq. (\ref{eq:15}) that
\begin{align}
\sigma^\mathrm{QKD}_{O_1^nC_1^nE}|_{\Omega_{\mathrm{PE}}}&=\mc{M}^n\circ\ldots\circ\mc{M}^1(\rho_{E_0})|_{\Omega_{\mathrm{PE}}}.\label{eq_32}
\end{align}

We note that as our channels $\mc{M}^i_{E_{i-1}\to O_iC_iE_i}$ contain no $R_i$ input or outputs, the non-signalling condition given by eq. (1.2) of \cite{metger2022generalised} is trivially satisfied.  In order to apply the GEAT, we further need to introduce the concept of a min-tradeoff function. A \emph{min-tradeoff function} $f:\mc{P}_{\mc{C}}\to\mathbb{R}$ is defined as a real valued function on the distributions of the alphabet $\mathcal{C}$, which we have defined in eq. (\ref{def:alphabetC}), such that for all $i=1,\ldots,n$ it holds 
\begin{equation}\label{eq:MinTrade}
f(p)\leq \min_{\ket{\rho}\in\Sigma_i( p)}H(O_i|E_i\tilde{E}_{i-1})_{\mc{M}^i\(\rho_{E_{i-1}\tilde{E}_{i-1}}\)},
\end{equation}
where $\tilde{E}_{i-1}$ can be chosen isomorphic to $E_{i-1}$ (noting that for our choice of $\mc{M}^i$, $R_{i-1}$ is trivial), and we have defined
\begin{equation}\label{def:Sigmai}
\Sigma_i(p)=\left\{\ket{\rho}_{E_{i-1}\tilde{E}_{i-1}}\in\mc{H}(E_{i-1}\tilde{E}_{i-1}): \Tr\left[\pro{c}_{C_i}\otimes\id_{O_iE_i}\mc{M}^i\(\rho_{E_{i-1}\tilde{E}_{i-1}}\)\right]= p(c)\;\forall c\in\mc{C} \right\}.
\end{equation}
The state $\rho$ can be chosen pure by strong subadditivity as remarked in \cite{dupuis2016entropy}. We can now reformulate the GEAT \cite[Theorem 4.3]{metger2022generalised} in the following way:

\begin{proposition}\emph{\cite{metger2022generalised}}\label{EAT2alpha}
Let $n\in\mathbb{N}$, $\epsilon\in(0,1)$ and $a\in(1,3/2)$. Let further $p_0$ be given by eqs. (\ref{eq:p0PE},\ref{eq:p0PE2}) and let $\Omega_{\mathrm{PE}}$ be the event defined by eq. \eqref{eq:OmegaPE} for some $\delta^\mathrm{tol}>0$, and an affine  min-tradeoff function $f$ such that $f(p)=f^\mathrm{PE}(p)+\mathrm{const}$. Then it holds for any $a\in(1,3/2)$, 
\begin{equation}\label{eq39}
H^\epsilon_{\min}(O_1^n|C_1^nE_n)_{\mc{M}^n\circ\ldots\circ\mc{M}^1(\rho_{E_0})|_{\Omega_{\mathrm{PE}}}}\geq nf(p_0)-n\(\delta^\mathrm{tol}+\frac{a-1}{2-a}\frac{\ln2}{2}V^2+\(\frac{a-1}{2-a}\)^2K_a\)-\frac{\Xi(\epsilon)}{a-1} -\frac{a}{a-1}\log\frac{1}{\Pr[\Omega_{\mathrm{PE}}]},
\end{equation}
where $\Xi(x)=-\log \left(1-\sqrt{1-x^2}\right)$, and we have defined
\begin{align}
V&=\sqrt{\mathrm{Var}(f)+2}+\log(2d_O^2+1),\label{eq:V}\\
K_a&=\frac{(2-a)^3}{6(3-2a)^3\ln2}2^{\frac{a-1}{2-a}(2\log d_O+\max(f)-\min_\Sigma(f))}\ln^3\(2^{2\log d_O+\max(f)-\min_\Sigma(f)}+e^2\),\label{eq:Ka}
\end{align}
where $\max(f)=\max_{p\in\mathcal{P}_{\mathcal{C}}}f(p)$ and $\min_\Sigma(f)=\min_{p:{\Sigma}\neq\text{\O}}f(p)$ and $\mathrm{Var}(f)$ denotes the variance of $f$.

\end{proposition}

We note that, whereas \cite{metger2022generalised} does not explicitly use infinite dimensional Hilbert spaces for the involved quantum systems, the proof can be extended to include infinite dimensional Hilbert spaces \cite{vNGEAT}. We then obtain

\begin{theorem}{\bf(Soundness)}\label{FinteSizePhys}
Let  $n\in\mathbb{N}$. Let $\epsilon,\epsilon_{\mathrm{PE}},\epsilon_\mathrm{EC},\epsilon_\mathrm{PA}\in(0,1)$. Let $f$ be an affine min-tradeoff function of the form $f(p)=f^\mathrm{PE}(p)+\mathrm{const}$. Let $p_0$ be given by eqs. (\ref{eq:p0PE},\ref{eq:p0PE2}). Let $\delta^\mathrm{tol}>0$ and let $\Omega_\mathrm{PE}$, $\Omega_\mathrm{EC}$, and $\Omega_\mathrm{NA}$ be defined by eqs. (\ref{eq:OmegaPE},\ref{eq:OmegaEC},\ref{eq:OmegaNonAbort}). Let $\text{leak}_\text{EC}$ be the amount of information leaked during error correction. Then, if $\Pr_{\sigma^\mathrm{QKD}}[\Omega_\mathrm{PE}]\geq\epsilon_{\mathrm{PE}}$, for any $a\in(1,3/2)$, the protocol provides an $(3\epsilon+\epsilon_\mathrm{EC})$-sound key at rate $r=\ell/n$ with
\begin{align}
   r|_{\Omega_{\mathrm{NA}}}&\geq f(p_0) -\delta^\mathrm{tol}-\frac{a-1}{2-a}\frac{\ln2}{2}V^2-\(\frac{a-1}{2-a}\)^2K_a-\frac{1}{n}\left[\frac{\Xi(\epsilon)}{a-1}+\frac{a}{a-1}\log\frac{1}{ \epsilon_{\mathrm{PE}}}+\mathrm{leak}_\mathrm{EC}+2\log{\frac{1}{\epsilon_\mathrm{PA}}}\right]\label{eq:FiniteKeyRate},
\end{align}

\end{theorem}

\begin{proof}
By Lemma \ref{lemma10vN} and Proposition \ref{EAT2alpha}, 
\begin{align}
H^{\epsilon}_{\min}(O_1^n|\tilde{X}_1^nI_1^nE'_n)_{\sigma^\mathrm{QKD}|\Omega_\mathrm{PE}\land \Omega_{\mathrm{EC}}} &\geq
H^{\epsilon}_{\min}(O_1^n|\tilde{X}_1^nI_1^nE'_n)_{\sigma^\mathrm{QKD}|{ \Omega_{\mathrm{PE}}}}\\
&\geq nf(p_0)-n\(\delta^\mathrm{tol}+\frac{a-1}{2-a}\frac{\ln2}{2}V^2+\(\frac{a-1}{2-a}\)^2K_a\) \nonumber \\
& - \frac{\Xi(\epsilon)}{a-1}-\frac{a}{a-1}\log\frac{1}{\Pr_{\sigma^\mathrm{QKD}}[\Omega_{\mathrm{PE}}]},
\end{align}
Further, it holds by assumption that $\Pr_{\sigma^\mathrm{QKD}}[\Omega_{\mathrm{PE}}]\geq \epsilon_{\mathrm{PE}}$. Application of Proposition \ref{renner} completes the proof.
\end{proof}

\subsection{Completeness}\label{sec:compl}

We will now show the completeness of the protocol, i.e. we provide a lower bound on the probability that an honest application of the protocol does not get aborted $\Pr_{\sigma^\mathrm{QKD}}[\Omega_{\mathrm{NA}}]\big|_\mathrm{hon}$. This can be done along the lines of \cite{BPWFA23}. Abortion can occur after parameter estimation, or after  error correction. By the union bound it therefore holds
\begin{equation}
    1-\Pr[\Omega_{\mathrm{NA}}]_{\sigma^\mathrm{QKD}}\big|_\mathrm{hon}\leq1- \Pr[\Omega_\mathrm{PE}]_{\sigma^\mathrm{QKD}}\big|_\mathrm{hon}+1-\Pr[\Omega_\mathrm{EC}]_{\sigma^\mathrm{QKD}}\big|_\mathrm{hon}.
\end{equation}

Let us first consider parameter estimation. Let $p_0(x,z)$, a distribution over $(x,z) \in \tilde{\mathcal{C}}$ obtained in an honest implementation, which is of the form (\ref{eq:p0PE}) and let  $\freq_{c^n_1}$ be an observed frequency distribution. By the definition of the event $\Omega_\mathrm{PE}$, it holds 
\begin{align}
    &\Pr_{\sigma^\mathrm{QKD}}[\Omega_\mathrm{PE}]\big|_\mathrm{hon}
    \geq\Pr\left[\left|f^\mathrm{PE}(\freq_{c_1^n})-f^\mathrm{PE}(p_0)\right|\leq\delta^\mathrm{tol}\right]|_\mathrm{hon}.
\end{align}
In the case of an honest implementation, we can assume an iid structure. Hence the protocol can be described by $n$ independent multinoulli trials with parameter $p_0$. This allows us to apply the concentration presented in \cite{agrawal_optimistic_2023}, Proposition 2. To do so, let us define an ordering $(x,z)\to i$ for $i=1+z+6x$; 
we then set $\pi_i=p_0(x,z)$, $\hat{\pi}_i=\freq_{c_1^n}(x,z)$ and $h_i=h_{x,z}$ for $i\in\{1,\ldots,|\tilde{\mathcal{C}}|\}$. Further we define $\pi_{|\tilde{\mathcal{C}}|+1}:=1-\sum_{i=1}^{|\tilde{\mathcal{C}}|}\pi_i$, $\hat{\pi}_{|\tilde{\mathcal{C}}|+1}:=1-\sum_{i=1}^{|\tilde{\mathcal{C}}|}\hat{\pi}_i$ and $h_{|\tilde{\mathcal{C}}|+1}=0$. Let us also define
\begin{align}
&\gamma_i:=\frac{\pi_i\(1-\sum_{j=1}^i\pi_j\)}{1-\sum_{j=1}^{i-1}\pi_j},\label{eq:79a}\\
&c_i:=h_i-\frac{\sum_{j=i+1}^{|\tilde{\mathcal{C}}|+1}h_j\pi_j}{1-\sum_{j=1}^{i}\pi_j},\label{eq:79b}
\end{align}
for $i\in \{1,\ldots,|\tilde{\mathcal{C}}|\}$, as well as 
\begin{equation}\label{eq:79c}
D:=\max_{i,j\in\{1,\ldots,|\tilde{\mathcal{C}}|+1\}}|h_i-h_j|.
\end{equation}
Let now $\epsilon_\mathrm{PE}^c\in(0,1)$. By Proposition 2 of \cite{agrawal_optimistic_2023}, it then holds with probability $1-\epsilon_\mathrm{PE}^c$ that 
\begin{equation}
    \left|f^\mathrm{PE}(\freq_{c_1^n})-f^\mathrm{PE}(p_0)\right|=\left|(\hat{\pi}-\pi)^Th\right|\leq 2\sqrt{\log\(\frac{n}{\epsilon_\mathrm{PE}^c}\)\sum_{i=1}^{|\tilde{\mathcal{C}}|}\frac{\gamma_ic^2_i}{n}}+\frac{3D}{n}\log\frac{n}{\epsilon_\mathrm{PE}^c}.
\end{equation}
If we now choose the tolerance parameter $\delta^{\text{tol}}$ such that
\begin{equation}\label{eq:deltaPE}
   \delta^\mathrm{tol}\geq 2\sqrt{\log\(\frac{n}{\epsilon_\mathrm{PE}^c}\)\sum_{i=1}^{|\tilde{\mathcal{C}}|}\frac{\gamma_ic^2_i}{n}}+\frac{3D}{n}\log\frac{n}{\epsilon_\mathrm{PE}^c} ,
\end{equation}
we can show a completeness bound
\begin{equation}
\Pr\left[\left|f^\mathrm{PE}(\freq_{c_1^n})-f^\mathrm{PE}(p_0)\right|\leq\delta^\mathrm{tol}\right]\big|_\mathrm{hon}\geq1-\epsilon_\mathrm{PE}^c.
\end{equation}
As for the error correction, let $\epsilon_\mathrm{EC}^c\in(0,1)$, such that $1-\Pr[\Omega_\mathrm{EC}]|_\mathrm{hon}\leq \epsilon_\mathrm{EC}^c$, i.e. we assume that the error correction aborts with a probability of at most $\epsilon_\mathrm{EC}^c$. We thus obtain

\begin{theorem}\label{theo:complete}{\bf(Completeness)}
Let $\epsilon_\mathrm{PE}^c\in(0,1)$. Let $\Omega_\mathrm{PE}$ be as defined in eq. (\ref{eq:OmegaPE}), with $\delta^\mathrm{tol}$ as in eq. (\ref{eq:deltaPE}), respectively. Let further $\epsilon_\mathrm{EC}^c\in(0,1)$ be a suitable completeness parameter for the error correction protocol used. Then the protocol is $(\epsilon_\mathrm{PE}^c+\epsilon_\mathrm{EC}^c)$-complete, in the sense that $\Pr_{\sigma^\mathrm{QKD}}[\Omega_\mathrm{NA}]|_\mathrm{hon}\geq1-\epsilon_\mathrm{PE}^c-\epsilon_\mathrm{EC}^c$.
\end{theorem}


\subsection{The Min-Tradeoff Function}\label{sec:MinTrade}

We will now derive an affine min-tradeoff function of the form $f(p)=f^\mathrm{PE}(p)+\mathrm{const}$, with $f^\mathrm{PE}$ defined by eq. (\ref{eq:fPE})\footnote{In that case, it holds for all $c_1^n\in\Omega_\mathrm{PE}$ that $f(\freq_{c_1^n})\geq f(p_0)-\delta^\mathrm{tol}$.}. As in our protocol the number of key rounds is much larger than the number of test rounds, we make use of the infrequent sampling version of the GEAT. To do so, we decompose the channels $\mc{M}^i$ as follows
\begin{align}
\mc{M}^{i}(\cdot)&=p^\mathrm{K}\mc{M}^i_\mathrm{K}(\cdot)\otimes\pro{\perp}_{C_i}\otimes\pro{0}_{I_i}+(1-p^\mathrm{K})\mc{M}^i_\text{test}(\cdot)\otimes\pro{\perp}_{O_i}\otimes\pro{1}_{I_i}.
\end{align}
We then define a \emph{crossover min-tradeoff function} \cite{dupuis2019entropy} as a function  $g:\mc{P}_{\tilde{\mc{C}}}\to\mathbb{R}$ such that for all $i=1,\ldots,n$ it holds
\begin{equation}\label{eq:CrossMinTrade}
g(\tilde{p})\leq \min_{\ket{\rho}\in\tilde\Sigma_i(\tilde p)}H(O_i|E_i\tilde{E}_{i-1})_{\mc{M}^i\(\rho_{E_{i-1}\tilde{E}_{i-1}}\)},
\end{equation}
where $\tilde{E}_{i-1}$ can be chosen isomorphic to $E_{i-1}$, and where we have defined
\begin{equation}\label{Sigmatilde}
\tilde\Sigma_i(\tilde{p})=\left\{\ket{\rho_{E_{i-1}\tilde{E}_{i-1}}}\in\mc{H}(E_{i-1}\tilde{E}_{i-1}): \Tr\left[\pro{c}_{C_i}\id_{O_iE_i}\mc{M}^i_\text{test}\(\rho_{E_{i-1}\tilde{E}_{i-1}}\)\right]=\;\tilde{p}(c)\forall c\in\tilde{\mathcal{C}} \right\}.
\end{equation}
Note that this constraint set differs from the one given by eq. (\ref{def:Sigmai}) in that the constraints are only in terms of $\mc{M}^i_\text{test}$. According to Lemma V.5 of \cite{dupuis2019entropy}, we can transform a crossover min-tradeoff function $g$ into a min-tradeoff function $f$ as follows
\begin{align}
f(\delta_c)&=\max(g)+\frac{1}{1-p^\mathrm{K}}\(g(\delta_c)-\max(g)\)\;\forall c\in\tilde{\mc{C}},\label{def:f1}\\
f(\delta_{(\perp,\perp)})&=\max(g),\label{def:f2}
\end{align}
where $\delta_c$ is the distribution that is equal to $1$ for the value $c$ and $0$ for any other values. Let us also define $\max(g)=\max_{\tilde p\in\mc{P}_{\tilde{\mc{C}}}}g(\tilde p)$ and $\min(g)=\min_{\tilde p\in \mc{P}_{\tilde{\mc{C}}}}g(\tilde p)$. If $p$ is of the form $p(c)=(1-p^\mathrm{K})\tilde p(c)$ for $c\in\tilde{\mathcal{C}}$ and $p(\perp,\perp)=p^\mathrm{K}$, it holds $f((1-p^\mathrm{K})\tilde p)=g(\tilde p)$ for all $\tilde p\in\mc{P}_{\tilde{\mc{C}}}$. It can be shown that \cite{dupuis2019entropy}
\begin{align}
&\max(f)=\max(g)\label{e46}\\
&\min_\Sigma(f)\geq\min(g)\\
0\leq&\mathrm{Var}(f)\leq\frac{1}{1-p^\mathrm{K}}\(\max(g)-\min(g)\)^2\label{e48}.
\end{align}
We can now bound the higher order terms in eqs. (\ref{eq:V},\ref{eq:Ka}) in terms of $\max(g)$ and $\min(g)$,
\begin{align}
V\leq\tilde V&=\sqrt{\frac{1}{1-p^\mathrm{K}}\(\max(g)-\min(g)\)^2+2}+\log(2d_O^2+1),\label{eq:96}\\
K_a\leq\tilde K_a&=\frac{(2-a)^3}{6(3-2a)^3\ln2}2^{\frac{a-1}{2-a}(2\log d_O+\max(g)-\min(g))}\ln^3\(2^{2\log d_O+\max(g)-\min(g)}+e^2\).\label{eq:97}
\end{align}

Let now $i\in\{1,\ldots,n\}$ be an arbitrary round. In order to derive our crossover min-tradeoff function, we consider an entanglement-based version of a round of our QKD protocol, as defined by the source replacement scheme \cite{bennett1992quantum,Grosshans2003Virtual,curty2004entanglement,ferenczi2012symmetries}. Namely, in the entanglement-based picture of a round $i$ of the protocol, Alice prepares a pure state
\begin{equation}\label{eq:Initial}
\ket{\psi}_{A_iA'_i}=\frac{1}{2}\sum_{x=0}^3\ket{x}_{A_i}\ket{\varphi_x}_{A'_i},
\end{equation}
which does not depend on $i$. Alice then sends $A'_i$ to Bob via the noisy quantum channel $\mc{N}^i_{A'_i\to B_i}$, and keeps the $A_i$ subsystem. Alice then measures her $A_i$ subsystem in the computational basis. Namely, she projects onto $\{\pro{x}_{A_i}\}_{x=0}^3$, obtaining random variable $X_i$, whereas Bob measures his output $B_i$ as he would do in the prepare-and-measure protocol. The remainder of the round is the same as in the prepare-and-measure scenario.

Formally, the action of the channels $\mc{M}^{i}$, as described by eqs. (\ref{M}-\ref{Mtest}), can be described as action performed on an entangled state of the form (\ref{eq:Initial}). Namely,
\begin{align}
&\mc{M}_\mathrm{K}^i\(\rho_{E_{i-1}\tilde{E}_{i-1}}\)=\sum_{(x,z)=(0,0)}^{(3,3)}\Tr_{A_iB_i}\left[\pro{x}_{A_i}\otimes R^z_{B_i}\(\tilde{\mc{N}}^i_{{A'_i}E'_{i-1}\to {B_i}{E'_i}}\({\psi}_{A_iA'_i}\otimes\rho_{E_{i-1}\tilde{E}_{i-1}}\)\)\right]\otimes\pro{z}_{\hat{Z}_i},\label{Mkey2}\\
&\mc{M}_\mathrm{test}^i\(\rho_{E_{i-1}\tilde{E}_{i-1}}\)=\sum_{(x,z)\in \tilde{\mathcal{C}}}\Tr_{A_iB_i}\left[\pro{x}_{A_i}\otimes\tilde{R}^z_{B_i}\(\tilde{\mc{N}}^i_{A'_iE'_{i-1}\to B_iE'_i}\({\psi}_{A_iA'_i}\otimes\rho_{E_{i-1}\tilde{E}_{i-1}}\)\)\right]\otimes\pro{xz}_{\tilde{X}_i\tilde{Z}_i}\label{Mtest2},
\end{align}
and $\mc{M}^{i}\left(\rho_{E_{i-1}\tilde{E}_{i-1}}\right)$ follows from eq. (\ref{M}). Using (\ref{Mtest2}), we can now equivalently formulate the constraints of the set (\ref{Sigmatilde}) as
\begin{equation}
\Tr\left[\pro{x}_{A_i}\otimes\tilde{R}^z_{B_i}\(\tilde{\mc{N}}^i_{A'_iE'_{i-1}\to B_iE'_i}\({\psi}_{A_iA'_i}\otimes\rho_{E_{i-1}\tilde{E}_{i-1}}\)\)\right]=\tilde{p}(x,z) ,
\end{equation}
for all tuples $(x,z)\in \tilde{\mathcal{C}}$. Further, let us note that we can reformulate
\begin{equation}
    \mathcal{M}^i_{E_{i-1}\to \hat{Z}_iC_iE_i}(\cdot)=\id_{I_1^{i-1}C_1^{i-1}}\otimes\widetilde{\mc{M}}_{A_iB_i\to\hat{Z}_iC_iI_i}\circ\tilde{\mc{N}}^i_{A'_iE'_{i-1}\to B_iE'_i}\({\psi}_{A_iA'_i}\otimes\cdot\) ,
\end{equation}
where we have defined
the map $\widetilde{\mc{M}}_{AB\to\hat{Z}CI}$
\begin{align}
\widetilde{\mc{M}}(\cdot)&=p^\mathrm{K}\sum_{(x,z)=(0,0)}^{(3,3)}\Tr_{AB}\left[\pro{x}_{A}\otimes R^z_{B}(\cdot)\right]\otimes\pro{z}_{\hat{Z}}\otimes\pro{\perp}_{C}\otimes\pro{0}_{I}\nonumber\\
&+(1-p^\mathrm{K})\sum_{(x,z)\in \tilde{\mathcal{C}}}\Tr_{AB}\left[\pro{x}_{A}\otimes\tilde{R}^z_{B}(\cdot)\right]\otimes\pro{xz}_{C}\otimes\pro{\perp}_{\hat{Z}}\otimes\pro{1}_{I},\label{Mtilde}
\end{align}
which does not depend on $i$. In order to obtain a crossover-min-tradeoff function we can now relax the minimisation in (\ref{eq:CrossMinTrade}) to make it independent of Eve's attack map $\tilde{\mc{N}}^i_{A'_iE'_{i-1}\to B_iE'_i}$. Namely, we consider the relaxation 
\begin{equation}\label{74}
\min_{\ket{\omega}_{AB\hat{E}}\in\Sigma'_{\hat{E}}(\tilde p)} H(\hat{Z}|IC\hat{E})_{\widetilde{\mc{M}}_{AB\to\hat{Z}CI}(\omega_{AB\hat{E}})}\leq\min_{\ket{\rho}\in\tilde\Sigma_i(\tilde p)}H(\hat{Z}_i|E_i\tilde{E}_{i-1})_{\mc{M}^i\(\rho_{E_{i-1}\tilde{E}_{i-1}}\)},
\end{equation}
where we have defined $\hat{E}:=I_1^{i-1}C_1^{i-1}E'_i\tilde{E}_{i-1}$ and
\begin{equation}\label{SigmaEprime}
\Sigma'_{\hat{E}}(\tilde p)=\left\{\ket{\omega}_{AB\hat{E}}\in\mc{H}(AB\hat{E}):\Tr\left[\left(\pro{x}_{A}\otimes\tilde{R}^z_{B}\right)\omega_{AB}\right]=\tilde{p}(x,z)\;\forall(x,z)\in\tilde{\mathcal{C}}\land\omega_{A}=\frac{1}{4}\sum_{i,j=0}^3\<\varphi_j|\varphi_i\>\ket{i}\bra{j}_A\right\}
\end{equation}
represents the feasible set (i.e., the set of all states compatible with the constraints imposed by the statistical test). Let us note that the constraint on Alice's marginal $A$ in (\ref{SigmaEprime}) can be added since the $A$ system never leaves Alice's lab, i.e. is not affected by Eve's attack map. Moreover, since the entanglement of the state (\ref{eq:Initial}) cannot be certified from Alice and Bob's statistics alone, the constraint on Alice's marginal is necessary to obtain nontrivial results \cite{BPWFA23}. 
Following \cite{lin2019asymptotic,BPWFA23}, we now consider a coherent version of the map $\widetilde{M}$ defined in (\ref{Mtilde}), corresponding to a coherent version of a round of the protocol. In a coherent version of a round of the protocol, Alice and Bob's measurements are performed coherently, i.e. by an isometry writing the result into a quantum register, but not yet depahsing it. This results in a final pure state of the protocol, allowing us to apply \cite[Theorem 1]{coles2012unification}, which removes the dependence on Eve's subsystems. The coherent version of a $\widetilde{M}$ can be expressed in terms of a CP map $\mc{G}:AB\to AB\hat{Z}$, given by a single Kraus operator
\begin{equation}\label{def:G}
G=\mathbb{1}_A\otimes\sum_{z=0}^3\sqrt{R^z_B}\otimes\ket{z}_{\hat{Z}},
\end{equation}
where the region operators $R^z$ are defined by eq. (\ref{R}). We can then measure the $\hat{Z}$ register in the computational basis to dephase the $\hat{Z}$ register. This operation $\mc{Z}:\hat{Z}\to\hat{Z}$, also known as pinching, can be described in terms of the Kraus operators 
\begin{align}\label{def:Z}
&Z_j= \mathbb{1} \otimes \pro{j}_{\hat{Z}},
\end{align}
for $j\in\{0,\ldots,3\}$, and the identity is extended to all registers other than $\hat{Z}$. Using \cite[Theorem 1]{coles2012unification}, we can then obtain the following

\begin{lemma}\label{Lemma:DimReduct}
For all Hilbert spaces $\mc{H}_{\hat{E}}$ it holds
\begin{align}\label{eq:Optimization}
&\min_{\ket{\omega}\in\Sigma'_{\hat{E}}(\tilde p)}H(\hat{Z}|IC\hat{E})_{\widetilde{\mc{M}}(\omega)}= p^\mathrm{K}\min_{\omega\in\Sigma'(\tilde p)}D(\mc{G}(\omega_{AB})||\mc{Z}[\mc{G}(\omega_{AB})]),
\end{align}
where we have defined the set
\begin{equation}
\Sigma'(\tilde p)=\left\{\omega_{AB}\in \mc{D}(\mc{H}_{AB}):\Tr\left[\(\pro{x}_A\otimes \tilde{R}_B^z\) \omega_{AB} \right]=\tilde p(x,z)\;\forall(x,z)\in\tilde{\mathcal{C}}\land \omega_{A}=\frac{1}{4}\sum_{i,j=0}^3\<\varphi_j|\varphi_i\>\ket{i}\bra{j}_A \right\}, \label{eq:feasibleSet}
\end{equation}
which is independent of the reference system $\hat{E}$.
\end{lemma}

The proof of Lemma \ref{Lemma:DimReduct} can be found in \cite[Appendix A]{BPWFA23}. We can now define a function $g:\mc{P}_{\tilde{\mc{C}}}\to\mathbb{R}$ in terms of the r.h.s. of eq (\ref{eq:Optimization}), and note that it is convex.

\begin{proposition}\emph{\cite{winick2018reliable,liu2019device}} \label{lemma:conv}
For a given $0< p^\mathrm{K}\leq 1$, the function $g:\mc{P}_{\tilde{\mc{C}}}\to\mathbb{R}$,
\begin{equation}\label{eq:MinTradeOff2}
g(\tilde p)=p^\mathrm{K}\min_{\omega\in\Sigma'(\tilde p)}D(\mc{G}(\omega_{AB})||\mc{Z}[\mc{G}(\omega_{AB})])
\end{equation}
is convex.
\end{proposition}
For any $\tilde p\in\mathcal{P}_{\tilde{\mc{C}}}$, eq. (\ref{eq:MinTradeOff2}) is a convex optimisation, whose dependence on the distribution $\tilde p$ is included in the constraints. For any $\tilde  p\in\mc{P}_{\tilde{\mc{C}}}$, it takes the explicit form 
\begin{align}\label{eq:OriginalProgram}
\min_{\omega_{AB}}\ &D(\mc{G}(\omega_{AB})||\mc{Z}[\mc{G}(\omega_{AB})]),\\
\text{s.t.} \; &\omega_{AB}\geq0,\nonumber\\
&\forall (x,z) \in \tilde{\mathcal{C}}:\nonumber\\
&\;\Tr\left[\(\pro{x}_A\otimes \tilde{R}_B^z\) \omega_{AB} \right]=\tilde p(x,z),\nonumber\\
&\;\Tr_B[\omega_{AB}]=\rho_A\nonumber.
\end{align}
where  $\{\tilde{R}^z\}_z$ are region operators as defined in \eqref{eq:OperatorsPE}, and we wrote Alice's marginal as
\begin{equation}
    \rho_A = \frac{1}{4}\sum_{i,j=0}^3\<\varphi_j|\varphi_i\>\ket{i}\bra{j}_A.
\end{equation}
Note that this constraint on the marginal already implies the unit trace condition, hence we removed the redundant condition from \eqref{eq:OriginalProgram}. By the nature of duality, which roughly speaking shifts the constraints into the dual objective, the objective function of the dual problem corresponding to \eqref{eq:OriginalProgram} will explicitly depend on $\tilde p$ in an affine way---thus, we can use the dual feasible point, which always provides a reliable lower bound to the original optimisation, and derive a crossover min-tradeoff function that is affine with respect to the probabilities.

\subsection{Optimizing the Min-Tradeoff Function}\label{Subsec:Cutoff}

In this section we describe how we can transform the optimization problem given by eq. \eqref{eq:OriginalProgram} to obtain an affine min-tradeoff function. Note that this optimisation involves an infinite-dimensional space on  Bob's side. In order to solve this inconvenience for the numerics, we will introduce a cutoff in the number of photons, $N_c\in\mathbb{N}$, and solve the resulting optimisation for increasing values of $N_c$. As the obtained results numerically converge, that is, at some point hardly change when increasing $N_c$, we will make a cutoff assumption and assume that the solution obtained for a large enough $N_c$ is almost equal to the one obtained using an unbounded $N_c$. This convergence is a reasonable assumption, provided the separability of the underlying Hilbert space. 

In what follows, we first show how to lower bound \eqref{eq:OriginalProgram}. A common technique for that is the method based on the Frank-Wolfe algorithm introduced in \cite{winick2018reliable}, used for example in \cite{kanitschar2023finite,BPWFA23}. This method is, however, numerically unstable, and as a first-order method it often stagnates before approaching the actual minimum, making it hard to achieve high precision. Instead, we shall use a recently introduced technique \cite{lorente2024} that directly solves \eqref{eq:OriginalProgram} using conic optimization. This allows us to find a solution efficiently, reliably, and with high precision.

The principle of this technique is reformulating the relative entropy as a convex cone \cite{fawzi2023optimal}, and using an efficient implementation \cite{coey2023} of the recently discovered Skajaa and Ye's algorithm \cite{skajaa2015,papp2017} to optimize over it. While this allows for the solution of a wide class of optimization problems involving the relative entropy, it is not yet enough to solve \eqref{eq:OriginalProgram}. The difficulty is that this problem is not strictly feasible, as the matrix $\mc{G}(\omega_{AB})$ will necessarily have null eigenvalues, and strict feasibility is necessary for the algorithm to work reliably. We emphasize that this is not a limitation of this particular method, but a generic feature of convex optimization algorithms \cite{drusvyatskiy2017,hu2021robust}.\footnote{We also note that $\omega$ itself does not have null eigenvalues according to the nature of our numerical method. Namely, the numerical solver starts with an initial, full rank state and performs the optimization according to a barrier function, whose design effectively prevents the states from having non-positive eigenvalues.}

In order to solve this issue, it is necessary to perform facial reduction along the lines of \cite{hu2021robust}. That is, to reformulate the problem in terms of positive definite matrices only. In order to do that, first we use the identity \cite{winick2018reliable}
\begin{equation}
    D(\mc{G}(\omega)||\mc{Z}[\mc{G}(\omega)]) = -H(\mc{G}(\omega)) + H(\mc{Z}[\mc{G}(\omega)]),
\end{equation}
where $H$ is the von Neumann entropy. Then we note that $\mc{G}(.)$ is an isometry, and therefore $H(\mc{G}(\omega)) = H(\omega)$, and furthermore $\hat{\mathcal{Z}}(\omega) := \mc{Z}[\mc{G}(\omega)]$ is positive definite for any positive definite $\omega$, as the region operators \eqref{R} are full rank. This allows us to define
\begin{equation}
r(\omega) := - H(\omega) + H(\hat{\mathcal{Z}}(\omega)), \label{eq:newobj}
\end{equation}
which is equal to the objective of \eqref{eq:OriginalProgram} but is now formulated exclusively in terms of positive definite matrices. This makes the problem strictly feasible, but it is no longer a minimization of a relative entropy, so we cannot use the relative entropy cone anymore. 

This makes it necessary to define a new convex cone to represent \eqref{eq:newobj}, namely \begin{equation}\label{eq:newcone}
    \mathcal{K}^{\hat{\mc{Z}}} = \left\{(y,\omega) \in \mathbb{R}\times \mathbb{H}^{4(N_c+1)}_{+}: y \geq - H(\omega) + H(\hat{\mathcal{Z}}(\omega)) \right\},
\end{equation}
where $\mathbb{H}^d_{+}$ is the set of positive semidefinite matrices in $d$ dimensions. Using this cone, we rewrite \eqref{eq:OriginalProgram} as
\begin{align}
\min_{y,\omega_{AB}}\ & y,\label{eq:ConicProgram2}\\
\text{s.t.} \; &(y,\omega_{AB})\in \mathcal{K}^{\hat{\mc{Z}}},\;\nonumber\\
&\forall (x,z) \in \tilde{\mathcal{C}}:\nonumber\\
&\;\Tr\left[\left(\pro{x}_A\otimes \tilde{R}_B^z\right) \omega_{AB} \right]=\tilde p(x,z),\nonumber\\
&\Tr_B[\omega_{AB}]=\rho_A\nonumber.
\end{align}
This is a conic problem that we can optimize with the technique from \cite{lorente2024}, as $\mathcal{K}^{\hat{\mc{Z}}}$ is a particular case of the QKD cone implemented there. Note that the cone explicitly depends on the map $\hat{\mc{Z}}$, which allows us to exploit its structure, unlike the generic relative entropy cone. We use the fact that $\hat{\mc{Z}}(\omega)$ is necessarily block diagonal in $\hat{Z}$, which greatly reduces the size of the matrices that need to be processed and thus greatly improves performance.

Once a solution $\tilde{\omega}$ to \eqref{eq:ConicProgram2} has been found, we can proceed with a linearisation of the original program in order to build the min-tradeoff function. We note that our numerical method also provides a solution for the dual of \eqref{eq:ConicProgram2}, which can be used directly to build the min-tradeoff function. Nevertheless, we use a linearisation to explicitly illustrate the construction of the min-tradeoff function, while noting that this extra step does not induce a noticeable loss in the final secret key rate. For a given $\tilde  p\in\mc{P}_{\tilde{\mc{C}}}$, let $\omega_{\tilde{p}}^*\in\Sigma'(\tilde p)$ be the minimiser of (\ref{eq:OriginalProgram}) and $\tilde{\omega}\in\mc{D}(\mc{H}_{AB})$ our solution\footnote{Actually, this method is valid for any feasible primal point.}. Then,

\begin{align}
\frac{g(\tilde p)}{p^\mathrm{K}}=r(\omega_{\tilde{p}}^*)&\geq r(\tilde{\omega})+\Tr\left[(\omega_{\tilde{p}}^*-\tilde{\omega})\nabla r(\tilde{\omega})\right]\\
&\geq r(\tilde{\omega})-\Tr\left[\tilde{\omega}\nabla r(\tilde{\omega})\right]+\min_{\sigma\in\Sigma'(\tilde p)}\Tr\left[\sigma\nabla r(\tilde{\omega})\right]. \label{eq:GradMinimization}
\end{align}
The first inequality is due to the fact that $r$ is a convex, differentiable function over the convex set $\mc{D}(\mc{H}_{AB})$, hence it can be lower bounded by its first order Taylor expansion at $\tilde{\omega}$ (see e.g. \cite[p.69]{boyd2004convex}), and the last inequality is due to the fact that $\omega_{\tilde p}^*\in\Sigma'(\tilde p)$ is the primal optimal. On the other hand, we calculate the matrix gradient of the objective function
\begin{align}
\nabla r(\omega) = \log (\omega ) - \hat{\mathcal{Z}}^\dagger [\log\hat{\mathcal{Z}}(\omega)], \label{eq:MatrixGrad}
\end{align}
where we note that $\tilde{\mc{Z}}(\omega)>0$ whenever $\omega >0$, so the matrix gradient $\nabla r(\omega)$ exists. Besides, we note that
\begin{equation}
    r(\tilde{\omega})-\Tr\left[\tilde{\omega}\nabla r(\tilde{\omega})\right] = 0
\end{equation}
for any $\tilde{\omega}$ \cite{george2022finite}, so we can effectively remove such term in \eqref{eq:GradMinimization}. Now, for any $\tilde{\omega}$ and $\tilde p$, the optimisation problem in eq. (\ref{eq:GradMinimization}) is an instance of semidefinite program in standard form, explicitly given by
\begin{align}
\min_{\sigma_{AB}}\ &\Tr\left[\sigma\nabla r(\tilde{\omega})\right],\label{eq:SDP2}\\
\text{s.t.} \; &\sigma_{AB}\geq0,\;\nonumber\\
&\forall (x,z) \in \tilde{\mathcal{C}}:\nonumber\\
&\;\Tr\left[\left(\pro{x}_A\otimes \tilde{R}_B^z\right) \sigma_{AB} \right]=\tilde p(x,z),\nonumber\\
&\;\Tr_B[\sigma_{AB}]=\rho_A\nonumber.
\end{align}
Deriving an expression for the min-tradeoff function requires to find the dual of this minimisation, which involves an affine function that incorporates the probabilities. We proceed with this task by firstly reformulating the constraint on Alice's marginal as follows \cite{LinPhD}. Let $\{\Theta_A^i\}_{i=1}^{16}$ be a basis spanning the space of Hermitian matrices on $\mc{H}_A$, and let
\begin{equation}
\theta_i=\frac{1}{4}\sum_{x,y=0}^3\<\varphi_y|\varphi_x\>\Tr\left[\Theta_A^i \ket{x}\bra{y}_A \right]
\end{equation}
for $i \in \{1,\ldots,16\}$. Then, the constraint on Alice's marginal $\rho_A$ can be expressed as
\begin{equation}
\Tr\left[\left(\Theta^i_A\otimes\mathbb{1}_{B}\right)\sigma_{AB}\right]=\theta_i, \quad \forall i \in \{1,\ldots,16\}.
\end{equation}
However, we note that for $x \in \{0,\ldots,3\}$
\begin{equation}
    \sum_{z=0}^5 \pro{x}_A\otimes \tilde{R}^z_B = \pro{x}_A \otimes \mathbb{1}_B.
\end{equation}
Therefore, four of the matrices that form $\{\Theta_A^i\}_{i=1}^{16}$ (as well as the unit trace condition) are actually linearly dependent on the constraints formed by parameter estimation measurements, so we will only effectively use the subset $\{\Theta_A^i\}_{i=1}^{12}$ for the remaining derivations. Let us now denote the dual problem of \eqref{eq:SDP2} as
\begin{equation}\label{eq:maximization}
\max_{\nu,\kappa \in\Sigma^*_{\tilde{\omega}}}\ell_{\tilde p}(\nu,\kappa)
\end{equation}
where $\nu=(\nu_{0,0},\ldots,\nu_{3,5}) \in \mathbb{R}^{|\tilde{\mc{C}}|}$, $\kappa = (\kappa_1,\ldots\kappa_{12})\in\mathbb{R}^{12}$ and the dual objective is explicitly  
\begin{equation}
\ell_{\tilde p}(\nu,\kappa)=\sum_{(x,z)\in\tilde{\mc{C}}}\nu_{xz}\tilde p(x,z) +\sum_{i=1}^{12}\kappa_{i}\theta_i ,
\end{equation}
which is affine in  the probabilities $\tilde p$. The dual feasible set $\Sigma^*_{\tilde{\omega}}$ is defined as
\begin{equation}
\Sigma^*_{\tilde{\omega}}=\left\{(\nu,\kappa)\in \mathbb{R}^{|\tilde{\mc{C}}|}\times \mathbb{R}^{12} :
\nabla r(\tilde{\omega})-\sum_{(x,z)\in\tilde{\mc{C}}}\nu_{xz}\(\pro{x}_A\otimes \tilde{R}_B^z\) -\sum_{i=1}^{12}\kappa_{i}\(\Theta^i_A\otimes\mathbb{1}_{B}\)\geq0\right\} ,
\end{equation}
which does not depend on $\tilde p$. Then, using weak duality one can find
\begin{align}
g(\tilde p)=p^\mathrm{K}r(\omega^*_{\tilde p})&\geq p^\mathrm{K}\max_{\nu,\kappa \in\Sigma^*_{\tilde{\omega}}}\ell_{\tilde p}(\nu,\kappa) \\
&\geq p^\mathrm{K} \ell_{\tilde p}(\nu,\kappa)\\
&=:g_{\nu,\kappa,\tilde{\omega}}(\tilde{p}),\label{98}
\end{align}
for any $\tilde{\omega}\in\mc{D}(\mc{H}_{AB})$ and any $(\nu,\kappa)\in\Sigma^*_{\tilde{\omega}}$. We note that for any such choice of variables, the function $g_{\nu,\kappa,\tilde{\omega}}:\mc{P}_{\tilde{\mc{C}}}\to\mathbb{R}$ is an affine crossover min-tradeoff function.

We now describe how we can numerically obtain optimal (up to numerical precision) choices for our parameters $\tilde{\omega}_0$, $\nu_0$ and $\kappa_0$ in the crossover min-tradeoff function (\ref{98}) for a given distribution $\tilde{p}\in\mc{P}_{\tilde{\mc{C}}}$, which we obtain from simulating an honest implementation of  the protocol.  We note that it is possible to analytically confirm that our choices for said points are feasible, i.e. that  $\tilde{\omega}_0\in\mc{D}(\mc{H}_{AB})$ and $(\nu_0,\kappa_0)\in\Sigma^*_{\tilde{\omega}_0}$. Thus we can verify that the corresponding function $g_{\nu_0, \kappa_0,\tilde{\omega}_0}$ is indeed a valid crossover min-tradeoff function. 

Once such two points have been calculated after solving \eqref{eq:ConicProgram2} and \eqref{eq:maximization}, we build
\begin{align}
g_0(\tilde{p})&:=g_{\nu_0,\kappa_0,\tilde{\omega}_0}(\tilde p)\nonumber\\
&={p^\mathrm{K}}\( \tilde{g}_0 +\sum_{{(x,z)\in\tilde{\mc{C}}}}\nu_{0,xz}\tilde p(x,z)\)\label{def:g},
\end{align}
with a constant (in the sense that it is independent of $\tilde {p}$)
\begin{equation}
\tilde{g}_0 := \sum_{i=1}^{12}{\kappa}_{0,i} \theta_i.
\end{equation}
 
In order to compute the higher order terms of the generalised EAT, we need to find $\max(g_{0})=\max_{\tilde p\in\mc{P}_{\tilde{\mc{C}}}}g_{0}(\tilde p)$ and $\min(g_{0})=\min_{\tilde p\in \mc{P}_{\tilde{\mc{C}}}}g_{0}(\tilde p)$. We note that, as $\mc{P}_{\tilde{\mc{C}}}$ is convex and $g_{0}$ is affine, we can restrict our considerations to the extreme points of $\mc{P}_{\tilde{\mc{C}}}$. Namely we get
\begin{align} \label{eq:Maxf}
&\max(g_{0})=p^\mathrm{K}\tilde{g}_0+p^\mathrm{K}\underbrace{\max\left\{\nu_{0,xz}:(x,z)\in\tilde{\mc{C}} \right\}}_{=:\max(\nu_0)},\\
&\min(g_{0})=p^\mathrm{K} \tilde{g}_0 +p^\mathrm{K}\underbrace{\min\left\{\nu_{0,xz}: (x,z)\in\tilde{\mc{C}}  \right\}}_{=:\min(\nu_0)}. \label{Ming}
\end{align}
In the case where the minimisers are non-positive and the maximisers are non-negative, we can upper bound 
\begin{align}\label{eq:128}
    \max(g_{0})-\min(g_{0})
    &\leq \max(\nu_0) - \min(\nu_0),
    \end{align}
which is independent of $p^\mathrm{K}$. Finally, we can introduce the min-tradeoff function induced by our crossover min-tradeoff function $g_0$, given by eq. (\ref{def:g}), via eqs. (\ref{def:f1},\ref{def:f2}).
\begin{align}
f(p)&=\sum_{c\in\tilde{\mathcal{C}}}p(c)\(\max(g_0)+\frac{1}{1-p^\mathrm{K}}[g_0(\delta_c)-\max(g_0)]\)+p(\perp,\perp)\max(g_0)\\
&=\max(g_0)+\sum_{c\in\tilde{\mathcal{C}}}\frac{p(c)}{1-p^\mathrm{K}}[g_0(\delta_c)-\max(g_0)]\\
&=\underbrace{p^\mathrm{K}\(\tilde{g}_0+\max(\nu_0)\)}_{=:\tilde{f}}+\sum_{(x,z)\in\tilde{\mc{C}}}\underbrace{p^\mathrm{K}\frac{\nu_{0,xz}-\max(\nu_0)}{1-p^\mathrm{K}}}_{=:h_{x,z}}p(x,z)\\
&=\tilde{f}+f^\mathrm{PE}(p),\label{eq:ourf}
\end{align}
where 
\begin{align}
&f^\mathrm{PE}(p)=\sum_{(x,z)\in\tilde{\mc{C}}} h_{x,z}p(x,z),\label{eq:ourfPE}
\end{align}
which is of the form of eq. (\ref{eq:fPE}), allowing us to use $f$ as a min-tradeoff function in Theorem \ref{FinteSizePhys}. When applying Theorem \ref{FinteSizePhys}, the term that goes into the bound on the key rate is $f(p_0)$, for a protocol-respecting probability distribution $p_0$. In such case, it holds
\begin{align}
    &p_0(x,z) = \left(1-p^\mathrm{K}\right)\tilde{p}_0(x,z) \quad \forall (x,z) \in \tilde{\mathcal{C}}, \\
    &\sum_{(x,z)\in \tilde{\mathcal{C}}} \tilde{p}_0(x,z) = 1.
\end{align}
Then, the min-tradeoff function can be reformulated in a more convenient form as follows.
\begin{align}
f(p_0)&=p^\mathrm{K}\(\tilde{g}_0+\max(\nu_0)+ \sum_{(x,z)\in\tilde{\mc{C}} }[\nu_{0,xz}-\max(\nu_0)]\tilde{p}_0(x,z)\)\\
&=p^\mathrm{K}\(\tilde{g}_0+ \sum_{(x,z)\in\tilde{\mc{C}}}\nu_{0,xz} \tilde{p}_0(x,z)\)\label{146}\\
&=g_0(\tilde{p}_0). \label{eq:FinalMinTradeoff}
\end{align}
In summary, we have derived an affine min-tradeoff function which depends on the following parameters: A state $\tilde{\omega}_0\in\mc{D}(\mc{H}_{AB})$ and the dual feasible points $(\nu_0,\kappa_0)\in\Sigma^*_{\tilde{\omega}_0}$. In order to find optimal values for the parameters and the corresponding key rates in a realistic implementation, we will employ the numerical method of \cite{lorente2024} in the following section.

\subsection{Assumptions for security}

Before moving to the computation of key rates, we would like to comment on the assumptions used in the derivation of our security proof. First of all, the considered protocol is formulated in a fully device-dependent manner and, therefore, assumes trusted state preparations by Alice and measurements by Bob, as in the vast majority of previous CVQKD security proofs. In our case, we also use two additional assumptions: \emph{(i) a bounded-energy assumption}, where Eve’s attack uses bounded, but arbitrary, energy; and \emph{(ii) a cut-off assumption}, saying that the asymptotic trade-off computed for large enough number of photons, $N_c$, is very close to the value for arbitrary $N_c$. This is not to be confused with the standard formulation of the cut-off assumption, usually presented as \emph{(i') Eve’s attack uses states only with a bounded number of photons $N_c$}.  In particular our security proof under assumptions (i) and (ii) also implies security under assumption (i'), but the opposite is not necessarily true.

\section{Numerical implementation and results} \label{subs:NumImplementation}
In order to verify that our approach produces non-trivial key rates in a realistic implementation, we simulate an experiment in which Alice and Bob are connected by an optical fibre of length $L$, characterised by an excess noise $\xi$ and transmittance $\eta = 10^{-\alpha_{\text{att}} L/10}$. For the numerical analysis, we take an attenuation of $\alpha_{\text{att}}=0.2$ dB/km and $\xi=1\%$, matching the current standards in telecom optical fibres. This provides us with a simulated distribution $p^\mathrm{s}_0$
\begin{align*}
     p^\mathrm{s}_0(x,z) &= \int_{0}^\delta \int_{\frac{\pi(2z-1)}{4}}^{\frac{\pi(2z+1)}{4}} \frac{\gamma \exp\left(\frac{-|\gamma e^{i \theta} - \sqrt{\eta}\phi_x|^2}{1 + \eta \xi/2}\right)}{4 \pi (1 + \eta \xi/2)}d\theta d\gamma  , && z \in \{ 0,\ldots,3\}, \, x \in \{0, \ldots,3\} \\
     p^\mathrm{s}_0(x,z) &= \int_{\delta}^\Delta \int_{0}^{2\pi} \frac{\gamma \exp\left(\frac{-|\gamma e^{i \theta} - \sqrt{\eta}\phi_x|^2}{1 + \eta \xi/2}\right)}{4 \pi (1 + \eta \xi/2)}d\theta d\gamma  ,  &&  z = 4, \, x \in \{0, \ldots,3\} \\
     p^\mathrm{s}_0(x,z) &= \int_{\Delta}^\infty\int_{0}^{2\pi} \frac{\gamma \exp\left(\frac{-|\gamma e^{i \theta} - \sqrt{\eta}\phi_x|^2}{1 + \eta \xi/2}\right)}{4 \pi (1 + \eta \xi/2)}d\theta d\gamma  , && z = 5, \, x \in \{ 0, \ldots,3\}
\end{align*}

Here, the different integration intervals represent the fragments of the phase space corresponding to the modulation expressed in \eqref{eq:discPE}, and $\phi_x \in \{i^x\alpha\}_{x=0
}^3$ are the coherent state amplitudes used by Alice with $\alpha \in \mathbb{R}$. The region operators for the constraints $\{\tilde{R}_B^z\}_z$ are given by \eqref{eq:OperatorsPE}, such that they require an integration in the same, corresponding intervals. As it will be the case for the states, the numerical analysis requires us to project the operators onto a Fock basis bounded by a cutoff $N_c$, where values $N_c\geq 10$ are typically quoted \cite{lin2019asymptotic,hu2021robust}. It is observed that the obtained key rates numerically converge when increasing the cutoff, which is expected since the states are bounded in their energies according to the attenuation at the channel (so that high photon states are unoccupied), and that the underlying Hilbert space is separable---such that any infinite dimensional state can be well approximated by an appropriate finite basis. Using then the inner product \cite{burnett_1998}
 \begin{equation}
     \left<\gamma e^{i \theta}\right.|\left. n\right> = \frac{\gamma^n e^{-\gamma^2 / 2}  e^{-i n \theta}}{\sqrt{n!}},
 \end{equation}
we can expand the coherent states into the Fock basis. This leads to the additional advantage that the integrals emerging in the representation of the region operators can be solved analytically, which improves the numerical performance of our method. With 
\begin{eqnarray}
    \Gamma(1+n,a) = \int_a^\infty x^{n} e^{-x} dx
\end{eqnarray}
as the upper incomplete Gamma function, we have for $z\in\{0,\ldots,3\}$
\begin{equation}
    \tilde{R}^z_B = \sum_{n,m=0}^{N_c} \ket{m}\bra{n} \frac{\Gamma(1+n) - \Gamma(1+n,\delta^2)}{2 \pi \sqrt{\Gamma(1+n)\Gamma(1+m)}} \int_{\frac{\pi(2z-1)}{4}}^{\frac{\pi(2z+1)}{4}} e^{i (m-n) \theta} d\theta,
\end{equation}
and we note that the angular integral has an analytical solution according to the value of $m-n$. In the case of $\tilde{R}^4_B, \tilde{R}^5_B$, the angular integral spans the whole range $[0,2\pi)$, making these operators diagonal in the Fock basis
\begin{subequations}\label{eq:DiagonalOperators}
\begin{align} 
    \tilde{R}^4_B &= \sum_{n=0}^{N_c} \pro{n} \frac{\Gamma(1+n,\delta^2) - \Gamma(1+n,\Delta^2)}{\Gamma(1+n)},  \\
     \tilde{R}^5_B &= \sum_{n=0}^{N_c} \pro{n} \frac{\Gamma(1+n,\Delta^2)}{\Gamma(1+n)}.
\end{align}
\end{subequations}
With all the components of the optimisation defined, we can now proceed with the numerical construction of the min-tradeoff function, and the finite key analysis according to the model described in \ref{Subsec:Cutoff}. As a first step, we set a cutoff $N_c$ which, as we will find, produces secret key rates which are not dependent on its particular value. On the other hand, $\delta, \Delta$ and $\alpha$ are design parameters that can be optimized in order to increase the secret key rate. With a choice for all the mentioned elements, we can now solve \eqref{eq:ConicProgram2}---to this end, we use the implementation of the QKD cone from \cite{lorente2024} with the solver Hypatia \cite{coey2022solving} in the programming language Julia via the interface JuMP \cite{LubinDunningIJOC,DunningHuchetteLubin2017}, which permits the optimisation over our cone \eqref{eq:newcone}. The codes here employed are available in the GitHub repository \cite{CPG_2024_2}.

\begin{figure}[h!]
\centerline{\input{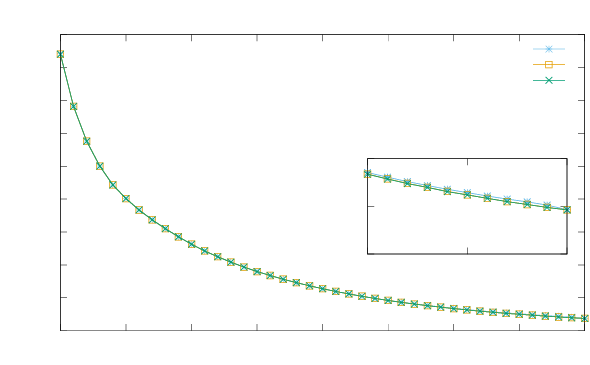}}
    \caption{Asymptotic secret keys for different values of the cutoff $N_c$, excess noise $\xi=1\%$ and error correction at the Shannon limit. The amplitude of the states was optimized at every distance, and $\Delta=5.0$, $\delta=2.0$.}
    \label{fig:Inset_Rates}
\end{figure}

Regarding the classical information leaked during error correction, we can assume an honest implementation of the protocol. For such case of an $\epsilon_\mathrm{EC}$-secure, robust error correction, 
\begin{equation} \label{eq:ErrorCorrection}
       \frac{1}{n}\mathrm{leak}_\mathrm{EC} \leq (1+f) H(\hat{Z}|\hat{X}) 
\end{equation}
where $\hat{X}$, $\hat{Z}$ represent one of the key string bits of Alice and Bob respectively after eliminating the $\perp$ symbol. The factor $f\geq 0$ represents the inefficiency of the error correction steps (with $f = 0$ for the Shannon limit, or perfect error correction) for a given \textit{hard decision} error correcting code \cite{leverrier2023information}, according to the fact that the information is already fully discretised at this step of the protocol. $H(\hat{Z}|\hat{X})$ can be computed numerically according to the distribution for Bob's heterodyne measurements adapted to the modulation for key rounds, namely
\begin{equation}
     p^\text{EC}_0 (x,z) = \int_0^\infty \int_{\frac{\pi (2z-1)}{4}}^{\frac{\pi(2z+1)}{4}} \frac{\gamma \exp\left(\frac{-|\gamma e^{i \theta} - \sqrt{\eta}\phi_x|^2}{1 + \eta \xi/2}\right)}{4 \pi (1 + \eta \xi/2)}d\theta d\gamma.
\end{equation}

\begin{figure}[h!]
\centerline{\input{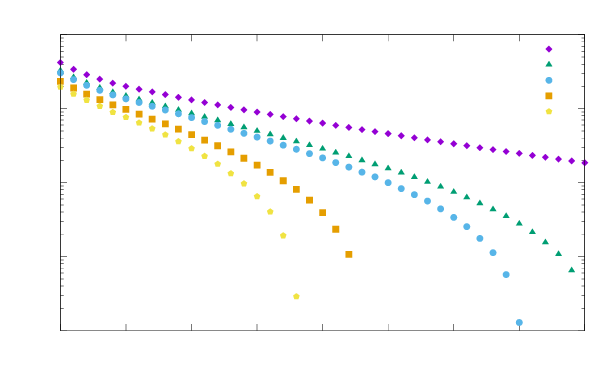}}
    \caption{Finite secret key rate according to \eqref{eq:FiniteKeyRate} for diverse block sizes $n$ under a cutoff $N_c=12$. The excess noise of the channel was taken as $\xi=1\%$, and the error correction at the Shannon limit. The amplitude of the states at every distance was taken to be the optimal one for $n\rightarrow \infty$, and we employed $N_c = 12$, $\Delta=5.0$, $\delta=2.0$ as well as $\epsilon=10^{-10}$, $\epsilon_\mathrm{PE}=10^{-10}$ and $\epsilon_\mathrm{PA}=10^{-10}$.}
    \label{fig:FRates}
\end{figure}


Our analysis starts with a choice for the cutoff parameter $N_c$, which we study in the case of the asymptotic regime. As it is provided in Figure \ref{fig:Inset_Rates}, the secret key rates do not noticeably change with respect to the value of $N_c$, and only exhibits slight differences at medium distances that vanish when $N_c>10$. This tight convergence is also a result of our numerical method based on conic programming, which provides improved results over previous approaches \cite[Figure 2]{BPWFA23}---where the Frank-Wolfe-based algorithm could only reach a limited precision without requiring an excessive overhead. According to these results, we set a cutoff $N_c=12$ for the remaining discussion, which ensures a proper balance in terms of reliability and time of execution.

Figure \ref{fig:FRates} provides diverse values for the secret key generation rate at the finite regime under the cutoff assumption, for different values $n$ of the total number of rounds as well as $\epsilon = \epsilon_\mathrm{PE} = \epsilon_\mathrm{PA} = 10^{-10}$, together with an optimized choice for the probability $p^\mathrm{K}$ according to the distance. Such curves were derived for $f=0\%$, and the amplitude for the states at every distance was taken to be the optimal in the case of $n \rightarrow \infty$. We observe, in particular, that for reasonable numbers of rounds such as $n=10^{8},10^{9}$ it is possible to perform our protocol while also obtaining high key generation rates in a regime of metropolitan distances (i.e., a few tens of kilometers). 

\begin{figure}[h!]
\centerline{\input{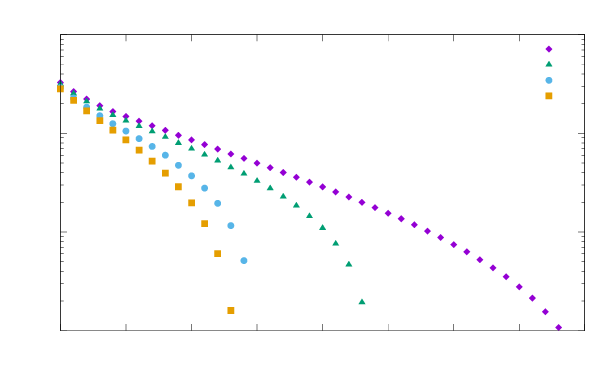}}
    \caption{Finite secret key rate for $n= 10^{10}$ rounds under diverse values of the error correction efficiency $f$. Any other parameters were taken to be the same as in Figure \ref{fig:FRates}.}
    \label{fig:FRates_Value_f}
\end{figure}

We conclude this analysis with a study of the error correction efficiency. This is illustrated by Figure \ref{fig:FRates_Value_f}, where the secret key rate at $n=10^{10}$ is represented according to different values of $f$ at the information leakage \eqref{eq:ErrorCorrection}. Here, we note that the signal to noise ratio for distances below $40$ km is of the order $10^{-1}$, for which efficient error correction methods exist such as LDPC or turbo codes. This makes our results represent a neat improvement over previous schemes, showing also the practicality of our approach---since they open the possibility of performing DM CVQKD under considerations of general security, accessible values of the error correction efficiency, as well as reduced numbers of rounds while keeping high key generation rates.

\subsection{Comparison to other works}

Let us now benchmark our outcomes against different results presented in literature. For instance, our approach only requires $n \sim 10^{8}$ rounds to perform at $15$ km according to Figure \ref{fig:FRates}, representing a neat improvement with respect to \cite[Figure 4]{BPWFA23}, which demands a minimum of $n \sim 10^{12}$. We make this enhancement explicit in Figure \ref{fig:FRates_GEAT_vs_EAT}---which also shows how eliminating the requirement of a virtual tomography simplifies the finite-size analysis, eliminating fluctuations present in the finite keys. Our results can also be compared to those provided by the protocols presented in \cite[Figures 2 and 3]{Matsuura2023} where secret keys at a maximum distance of 10 km are reported for $n=10^{12}$ and $\xi=1\%$; or \cite[Figure 3]{kanitschar2023finite} which also provides secret keys at 20 km with $n \sim 10^{9}$ albeit only in the case of collective security. However, for both cases we suggest a cautious comparison of the curves since said works employ a different scaling for the error correction (see e.g. \cite{leverrier2023information,denys2021explicit} for further information in this regard).

\begin{figure}[ht!]
\centerline{\input{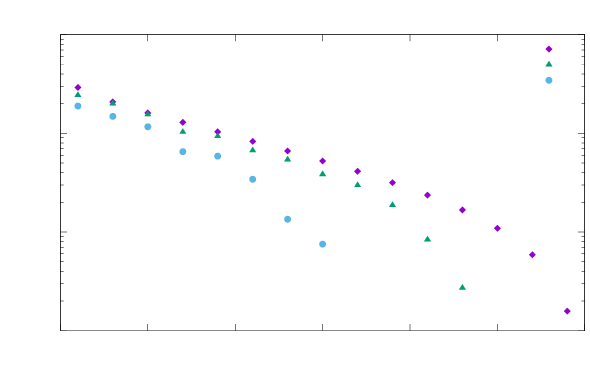}}
    \caption{Comparison of the secret key rates reported in \cite[Figure 4]{BPWFA23} and in our scenario under different values of $n$, together with $\xi = 1\%$, $N_c=12$ and $f=1\%$. For the curve derived under the GEAT, all other parameters were taken to be the same as in Figure \ref{fig:FRates}.}
    \label{fig:FRates_GEAT_vs_EAT}
\end{figure}

\subsection{Comparison to practical implementations}
The obtained improvements in terms of block sizes, however, come in principle with a caveat, namely, the method of GEAT~\cite{metger2022security} used in the security proof requires that Eve is only allowed to be in possession of one state at a time. Since this condition must be satisfied in every round, the repetition rate of the protocol clearly slows down, to enforce that there is only one signal in the channel at a time. For instance, the speed of light in an optical fiber is $\tilde{c}\approx 2 c/3$, for $c$ the one in vacuum, which for distances of the order of 20 km implies a repetition rate of the order of 10 Khz. This is however not a particularly relevant limitation for current CVQKD implementations.  In fact, the overall execution time in practical CVQKD implementations is fundamentally limited by the postprocessing of the signals -- such that, as in the case of a recent work~\cite{piétri2024qosst}, it takes of the order of 3 minutes to perform the digital signal processing for every frame of $10^6$ rounds. This indicates, provided a full post-processing of all the frames in parallel, a repetition rate of $\nu \approx 5.6\times 10^3$ Hz, slower than the one imposed by the GEAT requirements. Thus, using the GEAT does not necessarily represent a limitation  provided current technological standards, as long as Alice and Bob perform a postprocessing of the measured signals while also carrying the preparation and measurement of the remaining ones. 

Figure \ref{fig:RealAndGEAT} provides an estimation of the secret key rates in terms of bits per pulse, where the repetition rate $\nu$ is fixed according to the maximal value achieved due to the postprocessing time of \cite{piétri2024qosst}, and the repetition rate induced by the sequential structure of the GEAT -- namely, $\tilde{\nu} = \tilde{c}/L$ where $L$ is the distance between Alice and Bob. The latter represents the ultimate rate provided  by our analysis when the GEAT sequential requirement is the main limitation, while the former gives the expected rates when applying our formalism to state-of-the art implementations.



\begin{figure}
    \centerline{\input{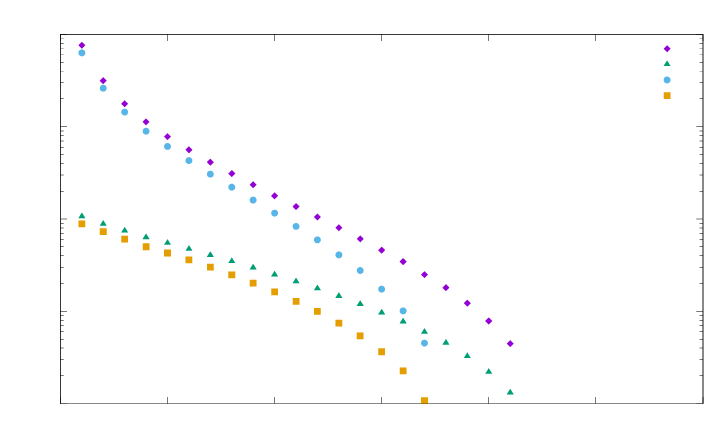}}
    \caption{Finite secret keys as bits per unit time according to different block sizes $n$ and repetition rates $\nu$ according to the postprocessing time \cite{piétri2024qosst}, and $\tilde{\nu} = \tilde{c}/L$ the sequential constraint imposed by the GEAT, with $\tilde{c}$ the speed of light in the fiber and $L$ the distance. Any other parameters were taken to be the same as in Figure \ref{fig:FRates}.}
    \label{fig:RealAndGEAT}
\end{figure}


\section{Discussion}
In this work, we provide a security proof for a DM CVQKD protocol in a scenario of general attacks. Albeit focused on four coherent states and fibre-based communications, we note that this approach can also be applied to any constellation of states, as well as different approaches as satellite-based communications provided the necessary changes in the model of the channel and actions performed by Alice and Bob. The application of the GEAT, in combination with the conic optimisation method of \cite{lorente2024}, allows us to distill positive key rates in the finite-size scenario of $n\sim10^8$ rounds, which represents a significant improvement in orders of magnitude over earlier works \cite{matsuura2021finite,Matsuura2023,yamano2024finite,BPWFA23}, which are also usually limited to a two-state constellation. In particular, block lengths of order $10^8$ are compatible with current experimental implementations, which approach the realm of GHz rates (see e.g. \cite{roumestan2022experimental}). 

This improvement in the finite-size regime can be partly attributed to the better performance of the new numerical method \cite{lorente2024}, as well as the use of the GEAT in combination with \cite{metger2022security} instead of the original EAT. In \cite{BPWFA23}, the EAT has to be combined with a virtual tomography of Alice's subsystem in order to take into account the fact that the protocol is prepare-and-measurement based, whereas the approach of \cite{metger2022security} allows for a direct application of the GEAT to said protocol. Since such virtual tomography requires not only extra steps in the security proof but also an expenditure in the key rate via statistical estimators, we thus achieve a simplified security framework as well as an increase in the finite key rate.

This enhanced performance comes with the structural limitations of the GEAT, where the repetition rate of the protocol is constrained by the requirement that Eve holds only one register at every step. As we observed from \cite{piétri2024qosst}, this is not critical for current experimental platforms, as long as the measurements are postprocessed during the preparation and measurement of the signals. However, future implementations may provide higher key rates by means of new techniques (e.g. code parallelization and precompilation, GPU-optimized performance), although said techniques typically come at a sharp cost and higher energy consumption. In this context, it also has been noted (see the discussion on page 3, and Corollary 6.1 of \cite{arqand2024generalized}) that it might be possible to lift the condition of Eve only having one subsystem at a time by combining the GEAT approach, which we use in this work, with the approach of \cite{HB24}. On similar grounds, the method of virtual tomography \cite{BPWFA23}, which could also be improved by using the numerical approach of \cite{lorente2024}, does not make any assumptions on the number of states Eve might hold at a time.



Another point to consider is the photon number cutoff. Whereas our security proof provided here already allows for infinite dimensional Hilbert spaces, our  numerical computation of tradeoff functions, or asymptotic key rates, fixes a maximal photon number of $N_c$, which is however increased until observing numerical convergence. There have been several results which overcome this assumption in the scenario of collective attacks \cite{upadhyaya2021dimension,lupo_quantum_2022,kanitschar2023finite} by splitting (\ref{eq:OriginalProgram}) into an optimisation over states on a finite Hilbert space and a correction term---using an energy test, which can be seen as an additional part of the parameter estimation, one can bound the probability of being outside the cutoff space with high probability. The bound on the probability then goes into the higher order term, as well as in the constraints of the finite-size optimisation. Considering the applicability of said model in the case of coherent attacks, we have realised that combining such an approach with the GEAT poses additional challenges in the setting of finite block sizes. When applying the GEAT, the outputs of the hypothetical energy tests have to be included into the alphabet $\tilde{\mathcal{C}}$, becoming a set of variables in the min-tradeoff function. As the dependence on these probabilities is not affine, the min-tradeoff function has to be linearised by taking a first order Taylor expansion. This is possible in principle, but did not provide us with positive key rates due to emerging numerical difficulties. We therefore leave the lifting of the photon number cutoff assumption for future work. In particular, the recently announced Rényi version of both the EAT and GEAT \cite{arqand2024generalized} looks promising in that respect, as it does not require the definition of an affine min-tradeoff function. Another possible avenue of reserach that could yield even better rates would be to use the techniques applied in \cite{george2022finite,kamin2024} to optimise the min-tradeoff function.

\section{Acknowledgments}

C.P.G thanks Marco Túlio Quintino for fruitful indications about numerical precision, and Yoann Piétri for suggestions about experimental aspects of CVQKD. We further thank Omar Fawzi, Min-Hsiu Hsieh, Lars Kamin, Florian Kanitschar, Bill Munro, Mizanur Rahaman, Gelo Noel Tabia, Ernest Tan, Toshihiko Sasaki and Shin-Ichiro Yamano for insightful discussions.  This work was supported by the ERC (AdG CERQUTE, grant agreement No. 834266), the AXA Chair in Quantum Information Science, Gobierno de España (Severo Ochoa CEX2019-000910-S, NextGen Quantum Communications and FUNQIP), Fundació Cellex, Fundació Mir-Puig, the EU (QSNP and Quantera Veriqtas), the Generalitat de Catalunya (CERCA program and the postdoctoral fellowship programme Beatriu de Pin\'{o}s), European Union's Horizon 2020 research and innovation programme under grant agreement No. 801370 (2019 BP 00097) within the  Marie Sklodowska-Curie Programme. The research of M.A. was supported by the European Union--Next Generation UE/MICIU/Plan de Recuperación, Transformación y Resiliencia/Junta de Castilla y León, and by the Spanish Agencia Estatal de Investigación, Grant No. RYC2023-044074-I.


\bibliography{CV}

\end{document}